\documentclass[footinbib,aps,amsmath,amssymb,twocolumn,floatfix,pra,superscriptaddress,longbibliography]{revtex4-1}

\usepackage{bm}
\usepackage{ntheorem}

\usepackage[colorlinks]{hyperref}
\hypersetup{%
	plainpages=true,
	breaklinks=true,
	hypertexnames=false,
	pageanchor=true,
	colorlinks=true,
	linkcolor={blue},
	citecolor={blue},
	urlcolor={blue},
	anchorcolor={black}
}

\newtheorem{theorem}{Theorem}
\newtheorem*{proof}{Proof}
\newtheorem{lemma}{Lemma}
\newtheorem{problem}{Problem}

\begin{document}


\title{Quantum Algorithm for Solving a Quadratic Nonlinear System of Equations}
\author{Cheng Xue}
\affiliation{CAS Key Laboratory of Quantum Information, University of Science and Technology of China, Hefei, Anhui, 230026, P. R. China}
\affiliation{CAS Center For Excellence in Quantum Information and Quantum Physics, University of Science and Technology of China, Hefei, Anhui, 230026, P. R. China}
\affiliation{Institute of Artificial Intelligence, Hefei Comprehensive National Science Center, Hefei, Anhui, 230026, P. R. China}
\author{Xiao-Fan Xu}
\affiliation{CAS Key Laboratory of Quantum Information, University of Science and Technology of China, Hefei, Anhui, 230026, P. R. China}
\affiliation{CAS Center For Excellence in Quantum Information and Quantum Physics, University of Science and Technology of China, Hefei, Anhui, 230026, P. R. China}
\author{Yu-Chun Wu}
 \email{wuyuchun@ustc.edu.cn}
 \affiliation{CAS Key Laboratory of Quantum Information, University of Science and Technology of China, Hefei, Anhui, 230026, P. R. China}
 \affiliation{CAS Center For Excellence in Quantum Information and Quantum Physics, University of Science and Technology of China, Hefei, Anhui, 230026, P. R. China}
 \affiliation{Hefei National Laboratory, Hefei, Anhui, 230088, P. R. China}
 \affiliation{Institute of Artificial Intelligence, Hefei Comprehensive National Science Center, Hefei, Anhui, 230026, P. R. China}
\author{Guo-Ping Guo}
 \affiliation{CAS Key Laboratory of Quantum Information, University of Science and Technology of China, Hefei, Anhui, 230026, P. R. China}
 \affiliation{CAS Center For Excellence in Quantum Information and Quantum Physics, University of Science and Technology of China, Hefei, Anhui, 230026, P. R. China}
  \affiliation{Hefei National Laboratory, Hefei, Anhui, 230088, P. R. China}
 \affiliation{Institute of Artificial Intelligence, Hefei Comprehensive National Science Center, Hefei, Anhui, 230026, P. R. China}
 \affiliation{Origin Quantum Computing Company Limited, Hefei, Anhui, 230026, P. R. China}

\begin{abstract}
  Solving a quadratic nonlinear system of equations (QNSE) is a fundamental, but important, task in nonlinear science.
  We propose an efficient quantum algorithm for solving $n$-dimensional QNSE.  Our algorithm embeds QNSE into a finite-dimensional system of linear equations using the homotopy perturbation method and a linearization technique; then we solve the linear equations with a quantum linear system solver and obtain a state which is $\epsilon$-close to the normalized exact solution of the QNSE with success probability $\Omega(1)$. The complexity of our algorithm is $O({\rm polylog}(n/\epsilon))$, which provides an exponential improvement over the optimal classical algorithm in dimension $n$, and the dependence on $\epsilon$ is almost optimal. Therefore, our algorithm exponentially accelerates the solution of QNSE and has wide applications in all kinds of nonlinear problems, contributing to the research progress of nonlinear science.

\end{abstract}

\maketitle


\section{Introduction}\label{sec1}

Nonlinear equations appear in many natural and social sciences, such as fluid dynamics  \cite{anderson1995computational}, biology \cite{hobbie2007intermediate}, atmospheric dynamics \cite{ghil2012topics}, and nonlinear vibration mechanics \cite{bishop2011mechanics}. By solving nonlinear equations, we understand various nonlinear phenomena, such as turbulence \cite{wilcox1998turbulence}, chaos \cite{alligood1996chaos}, and fractal \cite{s1992fractal}.
Most nonlinear equations have no or hardly solvable analytical solutions, such that many numerical methods have been developed \cite{dennis1996numerical}. When the dimension of the nonlinear equations is large, solving the nonlinear equations with classical computers requires too many computational resources and may exceed the ability of classical computers. There is great demand for developing more efficient algorithms for solving nonlinear equations. 

Quantum computing is a new model of computation which provides a quantum advantage in some specific problems \cite{shor1999polynomial,grover1996fast,harrow2009quantum}. 
A typical example is solving linear equations, where quantum computing provides exponential acceleration \cite{harrow2009quantum}.
There are already many quantum algorithms for solving various linear equations, such as systems of linear equations \cite{harrow2009quantum,childs2017quantum,subacsi2019quantum,xu2021variational} and linear differential equations \cite{clader2013preconditioned,berry2014highorder,Montanaro_2016,berry2017quantum,xin2020quantum,cao2013quantum,costa2019quantum,fillion2017quantum,engel2019quantum,arrazola2019quantum,linden2022quantum,childs2020quantumspectral,childs2020highprecision,nielsen_quantum_2002}. 
A natural idea is to use quantum computing to accelerate the solution of nonlinear equations. 
In recent years, some quantum algorithms for solving nonlinear differential equations have been proposed, such as nonlinear ordinary differential equations \cite{leyton2008quantum,lloyd2020quantum,liu2020efficient,kyriienko2021solving,xue2021quantumHomotopy,krovi2022improved,lin2022koopman,jin2022quantum}, the nonlinear Schr\"{o}dinger equation \cite{lubasch2020variational}, and Navier-stokes equations \cite{chen2022quantum,budinski2021quantum}. 

However, there are still few quantum algorithms for solving a system of nonlinear equations. 
A related algorithm proposed by Qian $et\ al$. \cite{qian2019quantum} is based on Grover's algorithm \cite{grover1996fast} and provides polynomial acceleration. 
Another related work is the quantum Newton's method proposed by Xue $et\ al$ \cite{xue2021quantum}. The quantum Newton's method is a quantum-classical hybrid algorithm constructed using quantum random access memory \cite{giovannetti2008quantum,giovannetti2008architectures,kerenidis2016quantum} and $l_{\infty}$ tomography \cite{Kerenidis2020Quantum}. Influenced by the sample complexity of $l_{\infty}$ tomography, the quantum advantage of the quantum Newton's method is verified only by numerical simulation.
Whether there are more effective quantum algorithms for solving a system of nonlinear equations requires further research.

In this paper, we focus on a special kind of system of nonlinear equations, the quadratic nonlinear system of equations (QNSE). QNSE appears in all kinds of nonlinear problems, such as quadratic programming \cite{mangasarian1994nonlinear}, nonlinear element analysis \cite{reddy2014introduction}, and nonlinear differential equations \cite{verhulst2006nonlinear}. In specific, QNSE often appears when solving quadratic nonlinear differential equations, including the Navier-Stokes equations in fluid dynamics \cite{anderson1995computational}, the logistic equation in biology \cite{hobbie2007intermediate}, and the Lorenz system in atmospheric dynamics \cite{ghil2012topics}. 
QNSE also appears when solving nonlinear differential equations in which the degree of nonlinear polynomials is higher than two because these differential equations can be approximate to quadratic nonlinear differential equations \cite{kerner1981universal}. Therefore, solving QNSE is a fundamental and important task, and the algorithm for accelerating the solution of QNSE has a wide range of applications.

We propose an effective quantum algorithm for solving $n$-dimensional QNSE.
In our algorithm, based on the homotopy perturbation method and a linearization technique, QNSE is embedded in a finite-dimensional system of linear equations. Then the condition number of the finite-dimensional system is optimized by splitting some subspaces of the finite-dimensional system.
Next, we solve the system of linear equations with a quantum linear system solver \cite{harrow2009quantum,childs2017quantum} and obtain a state which is $\epsilon$-close to the normalized exact solution of the QNSE with success probability $\Omega(1)$, where $\Omega$ represents an asymptotic notation \cite{cormen2022introduction}, which provides the asymptotic lower bound.
The complexity of our algorithm is $O({\rm polylog}(n/\epsilon))$, which provides an exponential improvement over the optimal classical algorithm in dimension $n$, and the dependence on $\epsilon$ is almost optimal. Our algorithm places some constraints on the QNSE; it is suitable when the linear component of the QNSE is well conditioned and is dominant in QNSE.

This paper is organized as follows. Sec. \ref{section-2} gives the definition of QNSE. The details of our algorithm are introduced in Sec. \ref{section-3}. Sec. \ref{section-4} gives the main result of our algorithm. Then we give some applications of our algorithm in Sec. \ref{section-5}. Finally, conclusions and discussions of our work are given in Sec. \ref{section-7}.

\section{Quadratic Nonlinear System of Equations}\label{section-2}

In this paper, QNSE is defined as
\begin{equation}\label{eq-0}
    F_0+F_1{\bm x}+F_2{\bm x^{\otimes 2}}=0,
\end{equation}
where $\bm{x}\in \mathbb{R}^n$ and $F_i\in \mathbb{R}^{n}\times \mathbb{R}^{n^i}$. We also have the following assumptions and definitions for Eq. (\ref{eq-0}):
\begin{itemize}
    \item [(1)] $F_1$ is invertible.
    \item [(2)] $F_1$ and $F_2$  are $s$-sparse.
    \item [(3)] Parameters $\alpha$, $\beta$, and $R$ are defined as
    \begin{equation}\label{eq-alpha}
        \begin{aligned}
            &\alpha:=\Vert F_1^{-1}\Vert \Vert F_0\Vert,\\
            &\beta:=\Vert F_1^{-1}\Vert \Vert F_2\Vert, \\
            &R:=\max\{4\alpha\beta, \Vert F_0\Vert\}.
        \end{aligned}
    \end{equation}
    In this paper, if not specifically noted otherwise, $\Vert\cdot\Vert=\Vert \cdot \Vert_2$.
    \item [(4)] Oracles $O_{F1}$ and $O_{F2}$ extract nonzero elements of $F_1$ and $F_2$, respectively. $O_{F1}$ consists of $O_{F11}$ and $O_{F12}$, and $O_{F2}$ consists of $O_{F21}$ and $O_{F22}$, which are written as 
    \begin{align}
        &O_{F11}|i\rangle|j\rangle=|i\rangle|f_1(i,j)\rangle,\\
        &O_{F12}|i\rangle|j\rangle|z\rangle=|i\rangle|j\rangle|z\oplus (F_1)_{i,j}\rangle,\\
        &O_{F21}|i\rangle|j\rangle=|i\rangle|f_2(i,j)\rangle,\\
        &O_{F22}|i\rangle|j\rangle|z\rangle=|i\rangle|j\rangle|z\oplus (F_2)_{i,j}\rangle,
    \end{align}
    where $f_1(i,j)$ and $f_2(i,j)$ represent the column index of the $j$th nonzero entry of the $i$th row of $F_1$ and $F_2$ respectively.
    \item [(5)] An oracle $O_{F0}$ is used to prepare the amplitude encoding of $F_0$, which is written as
    \begin{equation}
        O_{F0}\vert 0\rangle=\frac{1}{\Vert F_0 \Vert}\sum_{i=0}^{n-1}{F_{0,i}\vert i\rangle}.
    \end{equation} 
\end{itemize}

Formally, the problem to be solved is defined in Problem \ref{pro-1}.
\begin{problem}\label{pro-1}
    Consider QNSE defined in Eq. (\ref{eq-0}); $\bm{x}^{*}$ is an exact solution of Eq. (\ref{eq-0}), and $|x\rangle$ is the amplitude encoding of $\bm{x}^{*}$, which is written as 
    \begin{equation}
        |x\rangle=\frac{1}{\Vert \bm{x}^{*}\Vert}\sum_{i=0}^{n-1}{x^{*}_i|i\rangle}.
    \end{equation} 
    Given oracles $O_{F0}$, $O_{F1}$, and $O_{F2}$,  the goal is to output a state $|\bar{x}\rangle$ such that $\Vert |x\rangle-|\bar{x}\rangle \Vert \leq \epsilon$. 
\end{problem}

\section{Quantum Homotopy Perturbation Method}\label{section-3}

In this section, we introduce the overall process of our algorithm. The process contains three steps:
\begin{itemize}
    \item [(1)] Transform Eq. (\ref{eq-0}) into another kind of nonlinear equation with the homotopy perturbation method.
    \item [(2)] Embed the transformed nonlinear equations into a finite-dimensional system of linear equations, and solve the linear equations with a quantum linear system solver.
    \item [(3)] Measure some qubits of the output state of the quantum linear system solver and obtain the target state which represents a normalized approximate solution of Eq. (\ref{eq-0}).
\end{itemize}

The details of the whole process described above are introduced in the following three sections.

\subsection{Homotopy perturbation method}\label{set3-a}

The homotopy perturbation method is a classical method for solving nonlinear equations \cite{he1999homotopy,babolian2009some,chakraverty2019advanced}. The main process of the homotopy perturbation method for solving Eq. (\ref{eq-0}) is as follows.
We construct the homotopy $\nu(p):[0,1]\to \mathbb{R}^n$, which satisfies
\begin{equation}\label{nonlinear_eq}
    H(\nu,p)=F_0+F_1{\nu}+pF_2{\nu^{\otimes 2}}=0.
\end{equation}
With homotopy perturbation method, $\nu$ is written as
\begin{equation}\label{nu_exp}
    \nu=\nu_0+p\nu_1+p^2\nu_2+\dots +p^c\nu_c,
\end{equation}
where $c\in \mathbb{N}^{+}$.
Then substituting Eq. (\ref{nu_exp}) into Eq. (\ref{nonlinear_eq}) and equating the terms with identical powers of $p$, we have
\begin{equation}\label{linear_differential}
    \left\{
    \begin{aligned}
        &F_1\nu_i+F_0=0, &i=0,\\
        &F_1\nu_i+F_2\sum_{j=0}^{j=i-1}{\nu_{j}\otimes\nu_{i-1-j}}=0,& i=1,2,\dots,c.
    \end{aligned}\right.
\end{equation}
When $p=1$, $\nu$ is an approximate solution of Eq. (\ref{eq-0}), and we define
\begin{equation}\label{eq-26}
    \bm{\tilde{x}}=\nu(1)=\nu_0+\nu_1+\dots +\nu_c.
\end{equation}
The error bound of $\bm{\tilde{x}}$ is analyzed in Lemma \ref{error-threshold}. 

\begin{lemma}\label{error-threshold}
    Let $\bm{x}^{*}$ represent an exact solution of Eq. (\ref{eq-0}); 
    the approximate solution obtained with the homotopy perturbation method is $\tilde{\bm{x}}=\sum_{i=0}^{c}{\nu_i}$. When $R<1$ and $c\geq \log_{1/R}{\frac{\alpha}{\epsilon(1-R)}}$,
    where $\alpha$ and $R$ are defined in Eq. (\ref{eq-alpha}), $\tilde{\bm{x}}$ satisfies 
    \begin{equation}
        \Vert \bm{x}^{*}-\tilde{\bm{x}} \Vert\leq \epsilon.
    \end{equation}
\end{lemma}

The proof of Lemma \ref{error-threshold} is given in Appendix \ref{sec-lemma-error}. From Lemma \ref{error-threshold} we see that when $R<1$, Eq. (\ref{eq-26}) is convergent; the solution error $\epsilon$ decreases exponentially with $c$.

\subsection{Linear embedding}\label{linear-embed}

Then we embed Eq. (\ref{linear_differential}) into a finite-dimensional system of linear equations
\begin{equation}\label{eq-linear}
    A\bm{y}=\bm{b},
\end{equation}
where $\bm{y}=[\bm{y}_0,\bm{y}_1,\dots ,\bm{y}_c]$. The details of $\bm{y}$, $A$, and $\bm{b}$ are explained as follows.

First, ${\bm y}_0$ is defined as
\begin{equation}
    {\bm y}_0=[\nu_0+\nu_1+\dots + \nu_c],
\end{equation}
which means ${\bm y}_0$ contains an $n$-dimensional vector and ${\bm y}_{0,0}$ represents the approximate solution $\tilde{\bm x}$.

Then we substitute ${\bm y}_0$ into Eq. (\ref{linear_differential}) and get
\begin{equation}
    F_1{\bm y}_{0,0}+F_2\sum_{i=1}^{c}{\sum_{j=0}^{i-1}{\nu_j\otimes \nu_{i-1-j}}}=-F_0.
\end{equation}
We consider $\nu_j\otimes \nu_{i-1-j}$ to be an independent element and define it as a component of ${\bm y}_1$. Then ${\bm y}_1$ contains $c(c+1)/2$ $n^2$-dimensional vectors, which is written as 
\begin{equation}
    {\bm y}_1=[\nu_0\otimes \nu_0,\nu_0\otimes \nu_1,\nu_1\otimes \nu_0,\dots,\nu_{c-1}\otimes \nu_0].
\end{equation}
${\bm y}_2$ is generated from ${\bm y}_1$ in a similar way. In general, ${\bm y}_{i}$ is generated from ${\bm y}_{i-1}$. Repeating this process, we have ${\bm y}_c=[\nu_0^{\otimes{c+1}}]$, and ${\bm y}_{c}$ satisfies
\begin{equation}
    F_1^{\otimes c+1}{\bm y}_{c,0}=(-F_0)^{\otimes c+1},
\end{equation}
so ${\bm y}_c$ does not generate new elements. In summary, ${\bm y}_{i}$ can be written as
\begin{equation}
    \bm{y}_i=\left\{\begin{array}{lll}
            &[\nu_0+\nu_1+\dots + \nu_c], & i=0,\\ 
            &[\bm{y}_{i,0},\bm{y}_{i,1},\dots ,\bm{y}_{i,\beta_i-1}], & 1\leq i \leq c.
            \end{array}\right.
\end{equation}
$\beta_i$ denotes the number of terms in $\bm{y}_i$, and $\bm{y}_{i,j}$ represents the $j$th item of $\bm{y}_i$, which is written as $\bm{y}_{i,j}=\otimes_{k=0}^{i}{\nu_{a_{i,j,k}}}$; $a_{i,j,k}$ satisfies
\begin{equation}\label{def-aijk}
    \left\{\begin{array}{l}
        a_{i,j,k}\geq 0, \\
        i+1\leq\sum_{k=0}^{i}{(a_{i,j,k}+1)}\leq c+1.
        \end{array}\right.  
\end{equation}
By Eq. (\ref{def-aijk}), $\beta_i$ satisfies
\begin{equation}
    \beta_i=\left\{\begin{array}{ll}
        1, & i=0, \\
        \sum_{k=i}^{c}{\binom{k}{i}}, & 1\leq i \leq c.
        \end{array}\right.  
\end{equation}
We define 
\begin{equation}\label{eq-aij}
    \vec{a}_{i,j}=[a_{i,j,0},a_{i,j,1},\dots ,a_{i,j,i}].
\end{equation}
The mapping $(i,j)\to \vec{a}_{i,j}$ is a one-to-one mapping, the time complexity to compute this mapping or its reverse is $O({\rm poly}(c))$.

Next, we discuss the structure of matrix $A$ and vector ${\bm b}$. We set $\vec{a}_{i,0}=[0,0,\cdots,0]$; then $\bm{y}_{i,0}=\nu_0^{\otimes i+1}$ and satisfies
\begin{equation}\label{eq4}
    F_1^{\otimes i+1}\bm{y}_{i,0}=(-F_0)^{\otimes i+1}.
\end{equation} 
When $j\geq 1$, $\vec{a}_{i,j}$ has nonzero elements, and we assume the first nonzero element is $a_{i,j,k}$; then we have
\begin{align}\label{eq3}
    &I_n^{\otimes k}\otimes F_1\otimes I_n^{\otimes i-k}\bm{y}_{i,j}+\nu_{a_{i,j,0}}\otimes\dots\otimes \nu_{a_{i,j,k-1}} \notag\\
    &\quad \otimes F_2(\sum_{l=0}^{a_{i,j,k}-1}{\nu_{l}\otimes \nu_{a_{i,j,k}-1-l}})\otimes \dots \otimes \nu_{a_{i,j,i}}=0,
\end{align}
where $\nu_{a_{i,j,0}}\otimes\dots\otimes \nu_l\otimes\nu_{a_{i,j,k}-1-l}\otimes\dots\otimes \nu_{a_{i,j,i}}\in \bm{y}_{i+1}$. Therefore, Eq. (\ref{eq-linear}) can be expanded in the following form:
\begin{equation}\label{eq-14}
    \left(\begin{array}{cccc}
      A_{0,0} &A_{0,1} & &   \\
      & A_{1,1}  &\ddots &  \\
      & & \ddots &A_{c-1,c}  \\
      & & & A_{c,c}
      \end{array}\right)\left(\begin{array}{c}
          \bm{y}_{0} \\
          \bm{y}_{1} \\
          \vdots \\
          \bm{y}_{c}
          \end{array}\right)=
          \left(\begin{array}{c}
            \bm{b}_{0} \\
            \bm{b}_{1} \\
            \vdots \\
            \bm{b}_{c}
            \end{array}\right),
\end{equation}
where $A_{i,i}$ is an ($n^{i+1}\beta_i$)-dimensional square matrix and $A_{i,i+1}$ is an ($n^{i+1}\beta_i\times n^{ i+2}\beta_{i+1}$) dimensional matrix. The elements of $A_{i,i}$ and $A_{i,i+1}$ are determined by Eqs. (\ref{eq4}) and (\ref{eq3}). With Eqs. (\ref{eq4}) and (\ref{eq3}), the expression of ${\bm b}$ can also be obtained; in detail, the $i$th component of $\bm{b}$ is 
\begin{equation}\label{eq-25-0825}
    \bm{b}_i=[(-F_0)^{\otimes i+1},\bm{0},\dots,\bm{0}].
\end{equation}

From Eq. (\ref{eq4}), matrix $A$ contains the block matrix $F_1^{\otimes i},i=1,2, \cdots ,c+1$, which causes the condition number $\kappa_A$ of $A$ to increase exponentially with $c$. We optimize $\kappa_A$ by splitting $F_1^{\otimes i+1}\bm{y}_{i,0}=(-F_0)^{\otimes i+1}$ in Eq. (\ref{eq-14}) into 
\begin{equation}\label{eq-6}
    \left(\begin{array}{ccccc}
        B_{i,0} &I & & \\
        & B_{i,1} &\ddots & \\
        & & \ddots &I \\
        &  & & B_{i,i}
        \end{array}\right)\left(\begin{array}{c}
            \nu_0^{\otimes i+1} \\
            F_0\otimes \nu_0^{\otimes i} \\
            \vdots \\
            F_0^{\otimes i}\otimes \nu_0
            \end{array}\right)=
            \left(\begin{array}{c}
                \bm{0} \\
                \bm{0} \\
                \vdots \\
                -F_0^{\otimes i+1}
                \end{array}\right),
\end{equation}
where $B_{i,j}=I_n^{\otimes j}\otimes F_1 \otimes I_n^{\otimes i-j}$. As a result, $\bm{y}_{i,0}$ is redefined as
\begin{equation}\label{eq-6-1}
    \bm{y}_{i,0}=[\nu_0^{\otimes i+1},F_0\otimes \nu_0^{\otimes i},\dots ,F_0^{\otimes i}\otimes \nu_0],
\end{equation}
and the linear system $A\bm{y}=\bm{b}$ defined in Eq. (\ref{eq-linear}) is adjusted accordingly. In later sections, Eq. (\ref{eq-linear}) is defaulted to the adjusted linear system. The dimension of the linear system is
\begin{align}\label{dim-N}
    N&=\sum_{i=0}^{c}{n^{i+1}(\beta_i+i)}\notag\\
    &=(n+1)^{c+1}-1-sn+\frac{cn^{c}(n-1)-n^c+1}{(n-1)^2}n^2\notag\\
    &\approx (n+1)^{c+1}+cn^{c+1}.
\end{align}

Next, we solve Eq. (\ref{eq-linear}) with the quantum linear system solver proposed in  \cite{childs2017quantum}; operations $O_A$ and $O_b$ are required. $O_A$ consists of $O_{A1}$ and $O_{A2}$, which are defined as 
\begin{equation}\label{eq-oa}
    \begin{array}{c}
        O_{A1}|i\rangle|j\rangle=|i\rangle|f_a(i,j)\rangle,\\
        O_{A2}|i\rangle|j\rangle|z\rangle=|i\rangle|j\rangle|z\oplus A_{i,j}\rangle,
    \end{array}
\end{equation}
where $f_a(i,j)$ represents the column index of the $j$th nonzero entry of the $i$th row of $A$. $O_b$ is used to prepare the amplitude encoding of $\bm{b}$, which is defined as 
\begin{equation}\label{eq-ob}
    O_b|0\rangle=|b\rangle:=\frac{1}{\Vert \bm{b}\Vert}\sum_{i=0}^{c}{\sum_{j=0}^{\beta_i-1}{\Vert b_{i,j}\Vert|i,j\rangle|b_{i,j}\rangle}},
\end{equation}
where $|b_{i,j}\rangle$ is the amplitude encoding of $\bm{b}_{i,j}$; using Eqs. (\ref{eq-25-0825}) and (\ref{eq-6}), $b_{i,j}$ is written as 
\begin{equation}\label{eq-bij}
    \bm{b}_{i,j}=\left\{
        \begin{array}{cc}
            [\bm{0},\bm{0},\dots,-F_0^{\otimes i+1}], & 0\leq i\leq c, j=0,\\
            \bm{0}, & 0\leq i\leq c, 1\leq j\leq \beta_i-1.
        \end{array}
    \right.
\end{equation}
$O_A$ is constructed by querying $O_{F1}$ and $O_{F2}$, and $O_b$ is constructed by querying $O_{F0}$; the query complexity is given in Lemma \ref{oracle-construction}, and the proof of Lemma \ref{oracle-construction} is given in Appendix \ref{sec-A-oracle-construction}.

\begin{lemma}\label{oracle-construction}
    The operations $O_{A1}$ and $O_{A2}$ defined in Eq. (\ref{eq-oa}) can be constructed by querying $O_{F1}$ and $O_{F2}$ $O({\rm poly}(c))$ times; $O_b$ defined in Eq. (\ref{eq-ob}) can be constructed by querying $O_{F0}$ $O({\rm poly}(c))$ times.
\end{lemma}

After running the quantum linear system solver, we obtain the output state $\vert \bm{y}\rangle$, which is written as
\begin{equation}\label{def-vec-y}
    \vert \bm{y}\rangle=\sum_{i=0}^{c}{\sum_{j=0}^{\beta_i-1}{\vert i,j\rangle\vert \bm{y}_{i,j}\rangle}}.
\end{equation}

The query complexity of the quantum linear system solver is $O(s_A\kappa_A {\rm polylog}(s\kappa/\epsilon))$, and the sparsity of matrix $A$ is 
\begin{equation}
    s_A=\frac{c(c+1)}{2}s.
\end{equation}

$\kappa_A$ is decided by $c$, $F_1$, and $F_2$. As shown in Lemma \ref{A-condition}, when $c$, $F_1$, and $F_2$ satisfy some conditions, an upper bound of $\kappa_A$ is derived.

\begin{lemma}\label{A-condition}
    When $c$, $F_1$, and  $F_2$ satisfy
    \begin{equation}
        \Vert F_1^{-1} \Vert(1+(c+1)\Vert F_2 \Vert)<1,
    \end{equation}
    the condition number $\kappa_A $ of matrix $A$ satisfies 
    \begin{equation}
        \kappa_A \leq \frac{\kappa_{F_1} +1}{1-\Vert F_1^{-1} \Vert[1+(c+1)\Vert F_2 \Vert]},
    \end{equation}
    where $\kappa_{F_1}$ represents the condition number of $F_1$.
\end{lemma}

The proof of Lemma \ref{A-condition} is given in Appendix \ref{sec-A-condition}. From Lemma \ref{A-condition}, $\kappa_A$ does not increase exponentially with $c$ and mainly depends on $\kappa_{F_1}$.

\subsection{Measurement}

Finally, by measuring the first qubit register of $\vert \bm{y}\rangle$ to $\vert 0,0\rangle$, we obtain the target state $\vert \bm{\tilde{x}}\rangle=\vert \bm{y}_{0,0}\rangle$, a normalized approximate solution of Eq. (\ref{eq-0}).
This step is probabilistic; the success probability is
\begin{equation}\label{eq-prob}
    p=\frac{\Vert \bm{y}_{0,0} \Vert^2}{\Vert \bm{y} \Vert^2}.
\end{equation}
A lower bound of $p$ is given in Lemma \ref{success-theo}.

\begin{lemma}\label{success-theo}
    Consider the linear system $A\bm{y}=\bm{b}$ defined in Eq. (\ref{eq-linear}). Let $\eta^{'} =\Vert \bm{y}_{0,0} \Vert/R$. When $\Vert F_1^{-1}\Vert<1$ and $R < \sqrt{2}/2$, $p$ defined in Eq. (\ref{eq-prob}) satisfies 
    \begin{equation}
        p\geq \frac{(\eta^{'})^2(1-2R^2)}{(\eta^{'})^2(1-2R^2)+2}.
    \end{equation}
\end{lemma}

The proof of Lemma \ref{success-theo} is given in Appendix \ref{sec-lemma-success}. 
With Lemma \ref{success-theo}, $p$ mainly depends on $\eta^{'}$ and $1-2R^2$.

\section{Main Result}\label{section-4}

In this section, the main result of our work is given in Theorem \ref{main-theo}.

\begin{theorem}\label{main-theo}
    Given the QNSE $F_0+F_1\bm{x}+F_2\bm{x}^{\otimes 2}=0$ defined in Sec. \ref{section-2} with an exact solution $\bm{x}^{*}$, let $\eta=\Vert \bm{x}^{*}\Vert/R$, $c=\lceil\log_{\frac{1}{R}}{\frac{4\alpha}{\eta R\epsilon(1-R)}}\rceil$, and $G=\Vert F_1^{-1} \Vert[1+(c+1)\Vert F_2 \Vert]$,
    where $\alpha$ and $R$ are defined in Eq. (\ref{eq-alpha}).
    When
    \begin{equation}\label{eq-0407-02}
        \begin{array}{l}
            G<1, \\
            R <\sqrt{2}/2,
            \end{array} 
    \end{equation}
    there exists a quantum algorithm with success probability $\Omega(1)$ to obtain a normalized quantum state $\vert \bm{\bar{x}}\rangle$ satisfying $\Vert \bm{\bar{x}}-\frac{\bm{x}^{*}}{\Vert \bm{x}^{*}\Vert}\Vert \leq \epsilon $. The query complexity of the algorithm for the oracles of $F_0$, $F_1$, and $F_2$ is
    \begin{equation}
        O\left(\frac{\kappa_{F_1}  s{\rm polylog}\left(\frac{\Vert F_1\Vert}{\epsilon\eta (1-G)(1-2R^2)\Vert F_2 \Vert}\right)}{\eta(1-G)\sqrt{1-2R^2}}  \right).
    \end{equation}
    The gate complexity is the query complexity multiplied by a factor of ${\rm polylog}\left(\frac{n\Vert F_1\Vert}{\epsilon\eta (1-G)(1-2R^2)\Vert F_2 \Vert} \right)$.
\end{theorem}

The proof of Theorem \ref{main-theo} is given in Appendix \ref{sec-main-proof}. Here we give some discussion of Theorem \ref{main-theo}. The dependence on $n$ and $\epsilon$ of our algorithm is $O({\rm polylog}(n/\epsilon))$. Compared with the classical algorithm, our algorithm provides exponential acceleration. 

Exponential acceleration comes at the cost of stronger constraints on QNSE, which are listed in Eq. (\ref{eq-0407-02}). 
The role of $G<1$ is to limit the condition number of matrix $A$ so that it does not increase exponentially with $c$. $G$ can be written as 
\begin{equation}\label{eq-G-0418}
    G=\kappa_{F_1}\frac{1+(c+1)\Vert F_2\Vert}{\Vert F_1\Vert}.
\end{equation}
The other constraint $R <\sqrt{2}/2$ is used to bound the success probability of our algorithm. $R <\sqrt{2}/2$ also implies $R<1$, which is the convergence condition of our algorithm. $R$ mainly depends on $4\alpha\beta$ because when $\Vert F_0\Vert \geq \sqrt{2}/2$, $x$ can be rescaled to $\zeta x$ with a suitable constant $\zeta$ which keeps $4\alpha\beta$ unchanged and makes $\Vert F_0\Vert < \sqrt{2}/2$. Notice that 
\begin{equation}\label{eq-R-0418}
    4\alpha\beta=4\Vert F_1^{-1}\Vert^2\Vert F_2\Vert \Vert F_0\Vert=4\kappa_{F_1}^2\frac{\Vert F_0\Vert}{\Vert F_1\Vert}\frac{\Vert F_2\Vert }{\Vert F_1 \Vert}.
\end{equation}
From Eqs. (\ref{eq-G-0418}) and (\ref{eq-R-0418}), $G$ and $R$ measure the condition number $\kappa_{F_1}$ and the dominance of the linear component $F_1$ in QNSE from different aspects. 
When the linear component $F_1$ is well conditioned and is dominant in QNSE, $G$ and $R$ are small, and our algorithm is efficient. ``Well conditioned" means the condition number $\kappa_{F_1}$ is not large. ``Dominant" means that the strength of the linear component $\Vert F_1\Vert$ in QNSE is much larger than that of other components, i.e., $\Vert F_1 \Vert^2>>\Vert F_0\Vert\Vert F_2\Vert$ and  $\Vert F_1 \Vert>>1+(c+1)\Vert F_2\Vert$.

\section{Application}\label{section-5}

In this section we give two applications of our algorithm. 

The first application is a two-dimensional QNSE, which is defined as
\begin{equation}\label{obm12}
    \left\{\begin{array}{ll}
        8x_0-x_1-0.5x_0^2+0.5x_0x_1+0.2=0, \\
        -x_0+8x_1-0.5x_1^2+0.5x_1x_0-0.2=0.
        \end{array}\right.  
\end{equation}
The corresponding $F_0,F_1$, and $F_2$ are written as
\begin{align}
    &F_0=\left(\begin{array}{c}
        0.2 \\
        -0.2
        \end{array}\right),
    F_1=\left(\begin{array}{cc}
        8&-1 \\
        -1&8
        \end{array}\right),\notag\\
    &F_2=\left(\begin{array}{cccc}
        -0.5&0.5&0&0 \\
        0&0&0.5&-0.5
        \end{array}\right).
\end{align}

We set $c=2$ and compute the parameters $G$ and $R$,
\begin{align}
    & G\approx 7.24\times 10^{-1}<1,\notag\\
    & R\approx 2.82\times 10^{-1}<\sqrt{2}/2.
\end{align}
Then $\bm{y}=[\bm{y}_0,\bm{y}_1,\bm{y}_2]$, and each component of $\bm{y}$ is
\begin{align}
    &\bm{y}_0=[\nu_0+\nu_1+\nu_2],\notag\\
    &\bm{y}_1=[\nu_0\otimes\nu_0,F_0\otimes\nu_0,\nu_0\otimes\nu_1,\nu_1\otimes \nu_0],\notag\\
    &\bm{y}_2=[\nu_0^{\otimes 3},F_0\otimes \nu_0\otimes \nu_0,F_0\otimes F_0\otimes \nu_0].
\end{align}
The corresponding matrix $A$ and vector $\bm{b}$ are
\begin{widetext}
  \begin{eqnarray}
    A=\left(\begin{array}{cccccccc}
      F_1&F_2& &F_2&F_2 & & & \\
      &F_1\otimes I_2& I_4& & & & & \\
      & &I_2\otimes F_1 & & & & &\\
      & & &I_2\otimes F_1& &I_2\otimes F_2 & & \\
      & & & &F_1\otimes I_2& F_2\otimes I_2& & \\
      & & & & &F_1\otimes I_4& I_8& \\
      & & & & & &I_2\otimes F_1\otimes I_2 & I_8\\
      & & & & & & & I_4\otimes F_1
      \end{array}\right),
  \end{eqnarray}
\end{widetext}
\begin{equation}
    \bm{b}=[-F_0,\bm{0}_4,-F_0^{\otimes 2},\bm{0}_4,\bm{0}_4,\bm{0}_8,\bm{0}_8,-F_0^{\otimes 3}]^T,
\end{equation}
where $\bm{0}_4=[0,0,0,0]$ and $\bm{0}_8$ is similar. Then we solve $A\bm{y}=\bm{b}$ and obtain the component $\bm{y}_0$,
\begin{equation}
    \bm{\tilde{x}}=\bm{y}_0=[-2.2151849674\times 10^{-2},2.2292943149\times 10^{-2}].
\end{equation}
A numerical solution $\bm{x}^{*}$ of Eq. (\ref{obm12}) is obtained with the classical algorithm, which is written as 
\begin{equation}
    \bm{x}^{*}=[-2.2151848573\times 10^{-2}, 2.2292944259\times 10^{-2}].
\end{equation}
Then $\bm{\tilde{x}}$ satisfies
\begin{equation}
    \Vert \bm{x}^{*}-\bm{\tilde{x}} \Vert \approx 1.56\times 10^{-9}.
\end{equation}

The second application is a nonlinear boundary problem, which is defined as
\begin{equation}
    u_{xx}-u^2-\delta x^2=0, u(0)=0,\ u(1)=0,
\end{equation}
where $\delta=5\times 10^{-4}$. Then we discretize the equation in the following way:
\begin{equation}
    x_i=(i+1)h, i=0,1,\dots,n-1, h=1/(n+1).
\end{equation}
We define $u_i=u(x_i)$, and $u_{i,xx}$ is written as 
\begin{equation}
    u_{i,xx}=\frac{u_{i+1}-2u_i+u_{i-1}}{2h^2}.
\end{equation}
Then we have the QNSE 
\begin{align}\label{qnse-example3}
    &-u_{i+1}+2u_i-u_{i-1}+2h^2u_i^2+2h^2\delta x_i^2=0,\notag\\
    &i=0,1,\dots,n-1,
\end{align}
where $u_{-1}=1$ and $u_n=0$. Equation (\ref{qnse-example3}) can be represented as 
\begin{equation}\label{eq-0412-01}
    F_0+F_1{\bm u}+F_2{\bm u}^{\otimes 2}=0,
\end{equation}
where
\begin{equation}\label{eq-0401-01}
    F_0=2\delta h^2[x_0^2,x_1^2,\dots,x_{n-1}^2]^T,
\end{equation}
\begin{equation}
    F_1=\left(\begin{array}{cccc}
        2&-1& & \\
        -1&2&\ddots & \\
        &\ddots&\ddots & -1\\
        & & -1&2 
        \end{array}\right),
\end{equation}
\begin{equation}
    (F_2)_{i,j}=\left\{
        \begin{array}{cc}
            2h^2, & j=ni, i=0,1,\dots,n-1, \\
            0, & j\neq ni, i=0,1,\dots,n-1.
        \end{array}
    \right.
\end{equation}

We set $n=100$, $c=2$, and $\zeta=1200$ and rescale ${\bm u}$ to ${\bm w}=\zeta {\bm u}$; ${\bm \omega}$ satisfies 
\begin{equation}\label{eq-0407-v1}
    \zeta^2F_0+\zeta F_1 {\bm w}+F_2{\bm w}^{\otimes 2}=0.
\end{equation}
The rescaled QNSE satisfies 
\begin{align}
    & G\approx 9.01\times 10^{-1}<1,\notag\\
    & R\approx 6.27\times 10^{-1}<\sqrt{2}/2.
\end{align}

Then we solve Eq. (\ref{eq-0407-v1}) with our algorithm and have $\tilde{\bm{\omega}}$; then $\tilde{\bm{u}}=\frac{1}{\zeta}\tilde{\bm{\omega}}$. Equation (\ref{eq-0412-01}) is also solved with the classical algorithm, and the solution $\bm{u}^{*}$ is obtained; the error satisfies 
\begin{equation}
    \Vert \bm{u}^{*}-\bm{\tilde{u}} \Vert \approx 5.41\times 10^{-19}.
\end{equation}

\section{Conclusion and Discussion}\label{section-7}

In this paper, a quantum algorithm for solving QNSE was proposed. When focusing on the equation dimension $n$ and the solution error $\epsilon$, the complexity of our algorithm is $O(\rm{poly}[\log(n/\epsilon)])$. 
The process of solving QNSE with a classical computer is usually transformed into solving a system of linear equations iteratively \cite{rheinboldt1974methods}. The conjugate gradient method is a widely used linear system solver with complexity $O(ns\kappa)$ \cite{shewchuk1994introduction}, so compared with the classical algorithm, our algorithm provides an exponential improvement in dimension $n$. The dependence on $\epsilon$ of our algorithm complexity is $O(\rm{poly}(\log[1/\epsilon)])$, which is almost optimal.

In practice, with the outstanding performance in solving QNSE, our algorithm can be used as a subprogram to accelerate the process of computation in many nonlinear problems, such as quadratic programming \cite{mangasarian1994nonlinear} and quadratic nonlinear differential equations \cite{anderson1995computational,hobbie2007intermediate,ghil2012topics}.
Furthermore, our algorithm provides a different idea for solving a system of nonlinear equations with quantum computing, which could inspire more quantum algorithms for solving nonlinear equations.

Our algorithm considers only QNSE; it can also be generalized to a higher-order nonlinear system of equations. For example, an $m$-order nonlinear system of equations $\sum_{j=0}^{m}{F_j\bm{x}^{\otimes j}}=0$ can be transformed into a finite-dimensional system of linear equations $Ay=b$ through a process similar to that introduced in Sec. \ref{section-3}. The convergence condition, the expression of $Ay=b$, etc., depend on $F_j,\ j=0,1,\dots,m$. 

An open question is whether our algorithm can be optimized further. Notice that our algorithm has several restrictions on QNSE. As discussed in Sec. \ref{section-4}, the constraints of our algorithm are that the linear component $F_1$ is well conditioned and is dominant in QNSE, which limit the applications of our algorithm.
In the future we will consider optimizing the constraints of our algorithm, thereby expanding the applications of our algorithm.

\section*{Acknowledgments}
This work was supported by the National Natural Science Foundation of China (Grant No. 12034018), and the Innovation Program for Quantum Science and Technology No. 2021ZD0302300. 

\appendix

\section{Proof of Lemma \ref{error-threshold}}\label{sec-lemma-error}

\begin{proof}

The exact solution $\bm{x}^{*}$ can be written as
\begin{equation}
    \bm{x}^{*}=\sum_{i=0}^{\infty}{\nu_i}.
\end{equation}
Then
\begin{equation}
  \Vert \bm{x}^{*}-\tilde{\bm{x}} \Vert=\Vert \sum_{i=c+1}^{\infty}{\nu_i} \Vert.
\end{equation}
Next, we give an upper bound of $\Vert \nu_i \Vert$.
From Eq. (\ref{linear_differential}),
\begin{equation}
    \Vert \nu_0\Vert=\Vert F_1^{-1}F_0 \Vert \leq \alpha.
\end{equation}
We assume $\Vert \nu_i\Vert \leq \gamma_i \beta^i\alpha^{i+1}$, where $\gamma_i$ is a constant and $\gamma_0=1$. From Eq. (\ref{linear_differential}), $\gamma_i$ can be set as
\begin{equation}\label{catalan-list}
    \gamma_i=\sum_{j=0}^{i-1}{\gamma_j\gamma_{i-1-j}}, \gamma_0=1.
\end{equation}
Equation (\ref{catalan-list}) represents the Catalan series \cite{koshy2008catalan}, and $\gamma_i$ satisfies
\begin{equation}\label{eq-catalan}
    \gamma_i=\frac{1}{i+1} \binom{2i}{i}\approx\frac{4^i}{i^{3/2}\sqrt{\pi}}<4^i.
\end{equation}
Thus,
\begin{equation}
    \Vert \nu_i \Vert \leq \alpha (4\alpha\beta)^i\leq\alpha R^i.
\end{equation}
Considering $R<1$, we have
\begin{equation}
    \Vert \sum_{i=c+1}^{\infty}{\nu_i} \Vert\leq \sum_{i=c+1}^{\infty}{\Vert \nu_i \Vert}\leq \frac{\alpha R^{c+1}}{1-R}.
\end{equation}
Therefore, when $c\geq \log_{1/R}{\frac{\alpha}{\epsilon(1-R)}}$, $\tilde{\bm{x}}$ satisfies  $\Vert \bm{x}^{*}-\tilde{\bm{x}} \Vert\leq \epsilon$.

\end{proof}

\section{Proof of Lemma \ref{oracle-construction}}\label{sec-A-oracle-construction}

\begin{proof}
    (1) Construction process of $O_{A1}$ and $O_{A2}$.

    In fact, the construction process of $O_{A1}$ and $O_{A2}$ is the process to compute nonzero blocks of matrix $A$. We just need to compute the nonzero blocks of $A_{i,i}$ and $A_{i,i+1}$, $i\in \{0,1,\dots,c\}$.

    First, we regard $A_{i,i}$ as a $\beta_i$-dimensional diagonal block matrix; the $j$th block is written as $A_{(i,j),(i,j)}$. When $j=0$, $A_{(i,j),(i,j)}$ is the matrix defined in Eq. (\ref{eq-6}). When $j\in \{1,2,\dots,\beta_i-1\}$, we compute $\vec{a}_{i,j}$ defined in Eq. (\ref{eq-aij}) and find $k$, the label of $\vec{a}_{i,j}$'s first nonzero element. Then $A_{(i,j),(i,j)}=I_n^{\otimes k}\otimes F_1 \otimes I_n^{\otimes i-k}$. 

    Second, we regard $A_{i,i}$ as a ($\beta_i\times \beta_{i+1}$)-dimensional block matrix; one block is written as $A_{(i,j),(i+1,j^{'})}$. When $j=0$, $A_{(i,j),(i+1,j^{'})}=\bm{0}$ for $j^{'}\in \{0,1,\dots,\beta_{i+1}-1\}$. When $j\in \{1,2,\dots,\beta_i-1\}$, we compute $\vec{a}_{i,j}$ defined in Eq. (\ref{eq-aij}) and find $a_{i,j,k}$, the first nonzero element of $\vec{a}_{i,j}$. Then we compute the related $\vec{a}_{i+1,j^{'}}$ from Eq. (\ref{eq3}); the number of $\vec{a}_{i+1,j^{'}}$ is $a_{i,j,k}$. Next, we compute $j^{'}$ from $\vec{a}_{i+1,j^{'}}$; for these $j^{'}$, $A_{(i,j),(i+1,j^{'})}=I_n^{\otimes k}\otimes F_2 \otimes I_n^{\otimes i-k}$. 

    Therefore, for $i=0,1,\dots,c$, $j=0,1,\dots,\beta_i-1$, we can compute the nonzero $A_{(i,j),(i,j)}$ and $A_{(i,j),(i+1,j^{'})}$, the time complexity is $O({\rm poly}(c))$. The nonzero elements of $A_{(i,j),(i,j)}$ or $A_{(i,j),(i+1,j^{'})}$ can be extracted by querying $O_{F1}$ or $O_{F2}$. 

    By implementing the above computation process with a quantum circuit, $O_{A1}$ and $O_{A2}$ can be constructed directly. Notice that the expression of $A_{(i,j),(i,j)}$ or $A_{(i,j),(i+1,j^{'})}$ is related to $i$ and $k$, where $i,k\leq c$. $O_{F1}$ and $O_{F2}$ are required for different $i$ and $k$; therefore the query complexity of the whole process to $O_{F1}$ and $O_{F2}$ is $O({\rm poly}(c))$.

    (2)Construction process of $O_b$.

    Now we give the preparation process of $|b\rangle$ defined in Eq. (\ref{eq-ob}). From Eqs. (\ref{eq-ob}) and (\ref{eq-bij}), $|b\rangle$ can be simplified to
    \begin{equation}
        |b\rangle=\frac{1}{\Vert \bm{b}\Vert}\sum_{i=0}^{c}{\Vert F_0\Vert^{i+1}|i,0\rangle|b_{i,0}\rangle}.
    \end{equation}

    To prepare $|b\rangle$, we first prepare
    \begin{equation}
        |\psi\rangle=\frac{1}{M}\sum_{i=0}^{c}{\Vert F_0\Vert^{i+1}|i,0\rangle|0\rangle}.
    \end{equation}
    Notice that $|b_{i,0}\rangle$ can be prepared by querying $O_{F0}$ $i+1$ times; we define
    \begin{equation}
        U=\sum_{i=0}^{c}{|i,0\rangle\langle i,0|\otimes U_i},
    \end{equation}
    where $U_i|0\rangle=|b_{i,0}\rangle$. 
    Then we have
    \begin{equation}
        |b\rangle=U|\psi\rangle.
    \end{equation}
    The query complexity of this process to $O_{F0}$ is $O({\rm poly}(c))$.
\end{proof}

\section{Proof of Lemma \ref{A-condition}}\label{sec-A-condition}

We first give the following lemma.
\begin{lemma}\label{lemma-1}
    Given an $n$-dimensional invertible matrix $M$, $i\in \mathbb{N}^{+}$, matrix $P$ is defined as
    \begin{equation}
        P=\left(\begin{array}{ccccc}
            P_{0,0} &I & & & \\
            & P_{1,1} &I & & \\
            & & \ddots & \ddots & \\
            & & & P_{i-1,i-1} &I \\
            & & & & P_{i,i}
            \end{array}\right),
    \end{equation}
    where $P_{j,j}=I_n^{\otimes j}\otimes M \otimes I_n^{\otimes i-j}$.
    Then $P$ is invertible, and $P^{-1}$ satisfies
    \begin{equation}
        \Vert P^{-1}\Vert\leq \frac{\Vert M^{-1}\Vert (1-\Vert M^{-1} \Vert^{i+1})}{1-\Vert M^{-1} \Vert}.
    \end{equation}
\end{lemma}

\begin{proof}
    Let $Q=P^{-1}$ and consider $P$ and $Q$ to be $[(i+1)\times (i+1)]$-dimensional block matrices, $Q_{i,j}$ is written as
    \begin{equation}\label{eq-5}
        \begin{aligned}
            &Q_{j,j-k}=\bm{0}, &0\leq j\leq i, 0\leq j-k\leq i,\\
            &Q_{j,j}=P_{j,j}^{-1}, &0\leq j\leq i,\\
            &Q_{j,j+k}=-P_{j,j}^{-1}Q_{j+1,j+k},&0\leq j\leq i, 0\leq j+k\leq i.
        \end{aligned}
    \end{equation}
    From Eq. (\ref{eq-5}), $Q$ is an upper triangular block matrix. $Q$ is split into $Q=\sum_{k=0}^{i}{\Gamma_k}$, and $\Gamma_k$ contains the block $Q_{j,j+k}$. From Eq. (\ref{eq-5}) and the definition of $P_{j,j}$,
    \begin{equation}
        \Vert \Gamma_k \Vert = \Vert M^{-1} \Vert^{k+1}, k=0,1,\dots ,i.
    \end{equation}
    Therefore,
    \begin{equation}
        \Vert P^{-1}\Vert \leq \sum_{k=0}^{i}{\Vert \Gamma_k \Vert}\leq \frac{\Vert M^{-1}\Vert (1-\Vert M^{-1} \Vert^{i+1})}{1-\Vert M^{-1} \Vert}.
    \end{equation}
\end{proof}

Then the proof of Lemma \ref{A-condition} is as follows.

\begin{proof}
First, consider the upper bound of $\Vert A\Vert$. By the definition of $A_{i,i}$,
\begin{equation}\label{norm-aii}
    \Vert A_{i,i}\Vert \leq \Vert F_1\Vert+1,\quad i=0,1,\dots ,c.
\end{equation}
As $A_{0,1}A_{0,1}^T=\frac{c(c+1)}{2}F_2F_2^T$, $\Vert A_{0,1} \Vert $ satisfies
\begin{eqnarray}\label{norm-a01}
    \Vert A_{0,1} \Vert= \sqrt{\frac{c(c+1)}{2}}\Vert F_2 \Vert < (c+1)\Vert F_2 \Vert.
\end{eqnarray} 
When $i\geq 1$, consider $A_{i,i+1}$ to be $[(\beta_i+i)\times (\beta_{i+1}+i+1)]$-dimensional block matrix; each block has the form $I^{\otimes j}\otimes F_2\otimes I^{\otimes i-j}$, $ 0\leq j \leq i$. 
From the structure of $A_{i,i+1}$, it has no more than $c$ nonzero block matrices in each row or column, so $A_{i,i+1}$ can be split into at most $c$ matrices with at most one nonzero block matrix in each row or column, which leads to
\begin{equation}\label{norm-aii1}
    \Vert A_{i,i+1}\Vert \leq c\Vert F_2\Vert,\quad i\in[1,2,\dots ,c-1].
\end{equation}
Combining Eqs. (\ref{norm-aii}), (\ref{norm-a01}) and (\ref{norm-aii1}), we have
\begin{align}\label{norm-matA}
    \Vert A \Vert\leq& \Vert {\rm diag}(A_{0,0},A_{1,1},\dots ,A_{c,c}) \Vert\notag\\
    &\ + \Vert {\rm diag}(A_{0,1},A_{1,2},\dots ,A_{c-1,c}) \Vert \notag\\
    =&\max_{i=0}^{c}\{\Vert A_{i,i}\Vert \}+\max_{i=0}^{c-1}\{\Vert A_{i,i+1}\Vert \}\notag\\
    \leq& \Vert F_1\Vert+1+(c+1)\Vert F_2 \Vert.
\end{align}

Now we analyze the upper bound of $\Vert A^{-1} \Vert$. Let $B=A^{-1}$, and consider $A$ and $B$ to be $[(c+1)\times (c+1)]$-dimensional block matrices. The block $B_{i,j}$ satisfies
\begin{equation}\label{eq-17}
    \begin{aligned}
        &B_{i,i}=A_{i,i}^{-1},\quad 0\leq i\leq c,\\
        &B_{i,i+k}=-A_{i,i}^{-1}A_{i,i+1}B_{i+1,i+k},\\
        &\quad\quad 0\leq i\leq c-1,0\leq i+k\leq c. 
    \end{aligned}
\end{equation}
$B$ can be split by the distance from the diagonal,
\begin{equation}
    B=\sum_{k=0}^{c}{\Lambda_k}.
\end{equation}
$\Lambda_k$ contains the block $B_{i,i+k}$. $\Vert \Lambda_0 \Vert$ satisfies
\begin{equation}\label{eq-15}
    \Vert \Lambda_0 \Vert \leq \max_{0\leq i \leq c}{\Vert A_{i,i}^{-1}\Vert}.
\end{equation}
$A_{i,i}$ can be viewed as a block-diagonal matrix; the first block is the matrix described in Eq. (\ref{eq-6}), and other blocks are in the form $I_n^{\otimes j}\otimes F_1 \otimes I_n^{\otimes i-j}$. Therefore, from Lemma \ref{lemma-1},
\begin{equation}\label{eq-16}
    \Vert A_{i,i}^{-1} \Vert \leq \max\left\{\Vert F_1^{-1}\Vert,\frac{\Vert F_1^{-1}\Vert (1-\Vert F_1^{-1} \Vert^{i+1})}{1-\Vert F_1^{-1} \Vert} \right\}.
\end{equation}
Notice that $\Vert F_1^{-1}\Vert(1+(c+1)\Vert F_2\Vert)<1$ implies $\Vert F_1^{-1}\Vert<1$; we have
\begin{equation}
    \Vert \Lambda_0 \Vert \leq \frac{\Vert F_1^{-1}\Vert}{1-\Vert F_1^{-1} \Vert}.
\end{equation}
From Eqs. (\ref{norm-a01}) and (\ref{norm-aii1}), $\Vert A_{i,i+1} \Vert \leq (c+1)\Vert F_2 \Vert$. Therefore, 
\begin{equation}\label{eq-18}
    \Vert A_{i,i}^{-1}A_{i,i+1} \Vert \leq \frac{\Vert F_1^{-1}\Vert}{1-\Vert F_1^{-1} \Vert}(c+1)\Vert F_2 \Vert. 
\end{equation}
We define
\begin{equation}\label{eq-19}
    \zeta=\frac{\Vert F_1^{-1}\Vert}{1-\Vert F_1^{-1} \Vert}(c+1)\Vert F_2 \Vert.
\end{equation}
From Eqs. (\ref{eq-17}), (\ref{eq-18}), and (\ref{eq-19}),
\begin{equation}
    \Vert \Lambda_k \Vert \leq \zeta \Vert \Lambda_{k-1}\Vert,
\end{equation}
then
\begin{equation}
    \Vert B \Vert \leq \sum_{k=0}^{c}{\Vert \Lambda_k \Vert}\leq \frac{1-\zeta^{c+1}}{1-\zeta}\Vert \Lambda_0 \Vert.
\end{equation}
With the assumption $\zeta<1$, $\Vert A^{-1}\Vert$ satisfies
\begin{align}\label{norm-matA-inv}
    \Vert A^{-1} \Vert=\Vert B \Vert &\leq \frac{1}{1-\zeta}\times \frac{\Vert F_1^{-1}\Vert}{1-\Vert F_1^{-1} \Vert}\notag\\
    &= \frac{\Vert F_1^{-1} \Vert}{1-\Vert F_1^{-1} \Vert(1+(c+1)\Vert F_2 \Vert)}.
\end{align}
Therefore, from Eqs. (\ref{norm-matA}) and (\ref{norm-matA-inv}),
\begin{align}
    \kappa_A &=\Vert A\Vert\times \Vert A^{-1} \Vert\notag\\
    &\leq \frac{\kappa_{F_1} +\Vert F_1^{-1} \Vert[1+(c+1)\Vert F_2 \Vert]}{1-\Vert F_1^{-1} \Vert[1+(c+1)\Vert F_2 \Vert]}\notag\\
    &\leq \frac{\kappa_{F_1} +1}{1-\Vert F_1^{-1} \Vert[1+(c+1)\Vert F_2 \Vert]}.
\end{align}
\end{proof}

\section{Proof of Lemma \ref{success-theo}}\label{sec-lemma-success}

\begin{proof}
    First, reorder $\bm{y}$ as $\bm{y}=[\bm{y}_0^{'},\bm{y}_1^{'},\dots ,\bm{y}_c^{'},\bm{y}_{c+1}^{'}]$;
    $\bm{y}_i^{'}$ is divided into the following three cases:
    \begin{itemize}
      \item [(1)]When $i=0$, $\bm{y}_0^{'}=\bm{y}_0=[\sum_{j=0}^c{\nu_j}]$.
      \item [(2)]When $1\leq i\leq c$, the element of $\bm{y}_i^{'}$ is denoted as $\nu_{a_{i,0}^{'}}\otimes \nu_{a_{i,1}^{'}}\otimes \dots \otimes \nu_{a_{i,k}^{'}}$, which satisfies
      \begin{equation}\label{eq-success-2-1}
        \left\{
        \begin{aligned}
            &k\geq 1,\\
            &a_{i,j}^{'}\geq 0,\\
            &\sum_{j=0}^{k}{(a_{i,j}^{'}+1)}=i+1.
        \end{aligned}\right.   
      \end{equation}
      Therefore, the number of elements in $\bm{y}_i^{'}$ is $2^i-1$. Each item $\bm{y}_{i,j}^{'}$ in $\bm{y}_i^{'}$ satisfies
      \begin{equation}
          \Vert \bm{y}_{i,j}^{'} \Vert \leq \alpha^{k+1}R^{i-k}\leq \Vert F_0\Vert^{k+1}R^{i-k}\leq R^{i+1}.
      \end{equation}
      Then
      \begin{equation}\label{eq-20}
          \Vert \bm{y}_i^{'} \Vert^2 \leq (2^i-1)R^{2(i+1)}<(2R^2)^iR^2.
      \end{equation}
      Since $R<\sqrt{2}/2$,
      \begin{equation}\label{eq-8}
          \sum_{i=1}^{c}{\Vert \bm{y}_i^{'} \Vert^2}\leq \frac{2R^4}{1-2R^2}.
      \end{equation}
      \item [(3)]When $i=c+1$, $\bm{y}_{c+1}^{'}$ contains the elements generated in Eq. (\ref{eq-6-1}) except $\nu_0^{\otimes i+1}$ for $i=1,2,\dots,c$; the elements are in the form
      \begin{equation}\label{eq0802-1}
          F_0^{\otimes j}\otimes \nu_0^{\otimes i-j}, i=2,3,\dots,c+1, j=1,2,\dots ,i-1,
      \end{equation}
      which implies that $\bm{y}_{c+1}^{'}$ contains $c(c+1)/2$ elements. Each item in $\bm{y}_{c+1}^{'}$ satisfies $\Vert F_0^{\otimes j}\otimes \nu_0^{\otimes i-j} \Vert \leq \Vert F_0\Vert^{i}$. Then $\bm{y}_{c+1}^{'}$ satisfies
      \begin{align}\label{eq-9}
        \Vert \bm{y}_{c+1}^{'} \Vert^2 &\leq \Vert F_0\Vert^4\sum_{j=1}^{c}{j\Vert F_0\Vert^{2(j-1)}}\notag\\
        &< \frac{\Vert F_0\Vert^4}{(1-\Vert F_0\Vert^2)^2 }\notag\\
        &\leq \frac{R^4}{(1-R^2)^2 }.
      \end{align}
    \end{itemize}
    Therefore, with Eq. (\ref{eq-8}), Eq. (\ref{eq-9}), and $\Vert \bm{y}_0^{'} \Vert = \Vert \bm{y}_0 \Vert=\Vert \bm{y}_{0,0} \Vert = \eta^{'} R$, $p$ satisfies 
    \begin{align}
        p&=\frac{\Vert \bm{y}_0^{'} \Vert^2}{\Vert \bm{y}_0^{'} \Vert^2+\sum_{i=1}^{c}{\Vert \bm{y}_i^{'} \Vert^2}+\Vert \bm{y}_{c+1}^{'} \Vert^2}\notag\\
        &\geq \frac{(\eta^{'})^2 R^2}{(\eta^{'})^2 R^2+\frac{2R^4}{1-2R^2}+\frac{R^4}{(1-R^2)^2}}\notag\\
        &\geq \frac{(\eta^{'})^2}{(\eta^{'})^2+\frac{2R^2}{1-2R^2}+2}\notag\\
        &\geq \frac{(\eta^{'})^2(1-2R^2)}{(\eta^{'})^2(1-2R^2)+2}.
    \end{align}
\end{proof}

\section{Proof of Theorem \ref{main-theo}}\label{sec-main-proof}

To prove Theorem \ref{main-theo}, some lemmas in  \cite{berry2017quantum} are used, which are as follows.

\begin{lemma}\label{berry-lemma13}
    Let $\vert \psi\rangle$ and $\vert \varphi\rangle$ be two vectors such that $\Vert \vert \psi\rangle\Vert \geq \alpha>0$ and $\Vert \vert \psi\rangle-\vert \varphi\rangle\Vert \leq \beta$. Then 
    \begin{equation}
        \left\Vert \frac{\vert \psi\rangle}{\Vert \vert \psi\rangle \Vert}-\frac{\vert \varphi\rangle}{\Vert \vert \varphi\rangle \Vert} \right\Vert \leq \frac{2\beta}{\alpha}.
    \end{equation}
\end{lemma}

\begin{lemma}\label{berry-lemma14}
    Let $\vert \psi\rangle=\alpha\vert 0\rangle\vert \psi_0\rangle+\sqrt{1-\alpha^2}\vert 1\rangle\vert \psi_1\rangle$ and $\vert \varphi\rangle=\beta\vert 0\rangle\vert \varphi_0\rangle+\sqrt{1-\beta^2}\vert 1\rangle\vert \varphi_1\rangle$, where $\vert \psi_0\rangle$, $\vert \psi_1\rangle$, $\vert \varphi_0\rangle$, and $\vert \varphi_1\rangle$ are unit vectors and $\alpha,\beta\in [0,1]$. Suppose $\Vert \vert \psi\rangle-\vert \varphi\rangle \Vert \leq \delta<\alpha$. Then $\Vert \vert \psi_0\rangle-\vert \varphi_0\rangle \Vert \leq \frac{2\delta}{\alpha-\delta}$.
\end{lemma}

Then the proof of Theorem \ref{main-theo} is shown as follows.

\begin{proof}
    Construct the system of linear equations $A\bm{y}=\bm{b}$ defined in Eq. (\ref{eq-linear}), and define
    \begin{equation}\label{eq-27}
        \tilde{\bm{x}}=\bm{y}_{0}=\sum_{i=0}^{c}{\nu_i}.
    \end{equation}
    With Lemma \ref{error-threshold} and our choice of $c$,
    \begin{equation}\label{eq-20211122-01}
        \Vert \bm{x}^{*}-\tilde{\bm{x}} \Vert \leq \frac{\epsilon\eta R}{4}.
    \end{equation}
    With $\Vert \bm{x}^{*}\Vert= \eta R $, Eq. (\ref{eq-27}), and Lemma \ref{berry-lemma13},
    \begin{equation}\label{eq-main-01}
        \Vert \frac{\bm{x}^{*}}{\Vert \bm{x}^{*}\Vert}-\frac{\tilde{\bm{x}}}{\Vert \tilde{\bm{x}} \Vert} \Vert \leq \epsilon/2.
    \end{equation}
    The normalized state of $\bm{y}$ is written as
    \begin{equation}\label{eq-21}
        \vert \bar{\bm{y}}\rangle=\sum_{l=0}^{c}{\alpha_l\vert l\rangle\vert \bar{\bm{y}}_l\rangle}.
    \end{equation}
    From Eqs. (\ref{eq-27}) and (\ref{eq-21}),
    \begin{equation}\label{eq-277}
        \bar{\bm{y}}_0=\frac{\tilde{\bm{x}}}{\Vert \tilde{\bm{x}} \Vert}.
    \end{equation}
    Then $A\bm{y}=\bm{b}$ is solved with the quantum linear system solver proposed in  \cite{childs2017quantum}, and the output state is written as
    \begin{equation}\label{eq-22}
        \vert \bm{y}^{'}\rangle=\sum_{l=0}^{c}{\alpha_l^{'}\vert l\rangle\vert y_l^{'}\rangle}.
    \end{equation}
    Define $\eta^{'}=\Vert \tilde{\bm{x}}\Vert/R$; from Eq. (\ref{eq-20211122-01}), $\eta^{'}$ satisfies
    \begin{equation}\label{eq-100}
        \vert \eta-\eta^{'}\vert \leq \frac{\epsilon\eta}{4}.
    \end{equation}
    The solution error of the quantum linear system solver is set as
    \begin{equation}\label{delta-def}
        \delta=\frac{\eta^{'}\sqrt{5(1-2R^2)}}{30}\epsilon,
    \end{equation}
    From Theorem $5$ in  \cite{childs2017quantum},
    \begin{equation}\label{eq-23}
        \Vert \vert \bar{\bm{y}}\rangle-\vert \bm{y}^{'}\rangle\Vert \leq \delta.
    \end{equation}
    We define $|\bm{\bar{x}}\rangle$ as
    \begin{equation}\label{eq-225}
        |\bm{\bar{x}}\rangle:=|y_0^{'}\rangle.
    \end{equation}
    From Eqs. (\ref{eq-21}), (\ref{eq-22}), (\ref{eq-23}), and (\ref{eq-225}) and Lemma \ref{berry-lemma14}, $|\bm{\bar{x}}\rangle$ satisfies
    \begin{equation}\label{eq-main-02}
        \Vert \vert \bar{y}_0\rangle-|\bm{\bar{x}}\rangle \Vert\leq \frac{2\delta}{\alpha_0-\delta}.
    \end{equation}
    From Lemma \ref{success-theo},
    \begin{equation}\label{eq-24}
        \alpha_0\geq \sqrt{\frac{(\eta^{'})^2(1-2R^2)}{(\eta^{'})^2(1-2R^2)+4}}> \sqrt{\frac{(\eta^{'})^2(1-2R^2)}{5}}.
    \end{equation}
    The second inequality in the above equation defaults to $(\eta^{'})^2(1-2R^2)<1$. Combining Eqs. (\ref{delta-def}), (\ref{eq-main-02}), and (\ref{eq-24}), we have
    \begin{equation}\label{eq-11}
        \Vert \vert \bar{y}_0\rangle-|\bm{\bar{x}}\rangle \Vert\leq \epsilon/2.
    \end{equation}
    [Notice that when $(\eta^{'})^2(1-2R^2)\geq 1$, $\alpha_0\geq 1/\sqrt{5}$; then $\delta$ is set as $\delta=C\epsilon$, where $C<\sqrt{5}/20$, and Eq. (\ref{eq-11}) can also be derived.]
    
    From Eqs. (\ref{eq-27}), (\ref{eq-main-01}), (\ref{eq-277}), and (\ref{eq-11}), the solution error satisfies
    \begin{equation}
        \left\Vert \frac{\vert \bm{x}^{*}\rangle}{\Vert\bm{x}^{*}\Vert}-|\bm{\bar{x}}\rangle \right\Vert \leq \epsilon.
    \end{equation}

    The solution error influences the success probability of our algorithm. We default to $\epsilon<0.1$; then from Es. (\ref{delta-def}) and (\ref{eq-24}), $\delta$ satisfies
    \begin{equation}
        \delta\leq \frac{\epsilon}{6}\alpha_0<\frac{1}{60}\alpha_0.
    \end{equation}
    Then
    \begin{equation}
        \alpha_0^{'}\geq \alpha_0-\delta\geq \frac{59}{60}\alpha_0.
    \end{equation}
    With Eq. (\ref{eq-100}) and $\eta^{'}\geq (1-\epsilon/4)\eta$,
    \begin{align}\label{eq-25}
        p=&(\alpha_0^{'})^2\geq \left(\frac{59}{60}\right)^2 \frac{1}{5} (1-\epsilon/4)^2\eta^2(1-2R^2) \notag\\
        >& 0.18\eta^2(1-2R^2).
    \end{align}
 
    Next, we analyze our algorithm complexity. From Lemma \ref{oracle-construction}, $O_{A}$ defined in Eq. (\ref{eq-oa}) can be constructed by querying $O_{F1}$ and $O_{F2}$ $O(\rm{poly}(c))$ times; $O_b$ defined in Eq. (\ref{eq-ob}) can be constructed by querying $O_{F0}$ $O(\rm{poly}(c))$ times. From Eq. (\ref{dim-N}), the dimension of the matrix $A$ is $N\approx (n+1)^{c+1}+cn^{c+1}$.
    The sparsity of $A$ is $s_A =\frac{c(c+1)}{2}s$.
    From Lemma \ref{A-condition}, $\kappa_A \leq\frac{\kappa_{F_1} +1}{1-G}$.
    From Theorem $5$ in  \cite{childs2017quantum}, when solving the linear system $A\bm{y}=\bm{b}$, the query complexity of $O_A$ and $O_b$ is
    \begin{equation}\label{eq-12}
        O(s_A \kappa_A {\rm polylog}(s_A \kappa_A /\delta)).
    \end{equation}
    Substituting the relevant parameters into Eq. (\ref{eq-12}), the query complexity of $O_{F0}$, $O_{F1}$, and $O_{F2}$ is
    \begin{equation}
        O\left(\frac{\kappa_{F_1}  s}{1-G} {\rm polylog}\left(\frac{\Vert F_1\Vert}{\epsilon\eta (1-G)(1-2R^2)\Vert F_2 \Vert}\right) \right).
    \end{equation}
    From Eq. (\ref{eq-25}), the success probability of the algorithm is $O(\eta^2(1-2R^2))$. Using amplitude amplification \cite{brassard2002quantum}, we repeat the above procedure $O(\frac{1}{\eta\sqrt{1-2R^2}})$ times and obtain the state $\vert \bm{\bar{x}}\rangle$ with probability $\Omega(1)$. Therefore, the query complexity of our algorithm is 
    \begin{equation}
        O\left(\frac{\kappa_{F_1}  s\cdot{\rm polylog}\left(\frac{\Vert F_1\Vert}{\epsilon\eta (1-G)(1-2R^2)\Vert F_2 \Vert}\right)}{\eta(1-G)\sqrt{1-2R^2}}  \right),
    \end{equation}
    and the gate complexity is the query complexity multiplied by a factor of ${\rm polylog}\left(\frac{n\Vert F_1\Vert}{\epsilon\eta (1-G)(1-2R^2)\Vert F_2 \Vert}\right)$.
\end{proof}

\bibliography{qhpmNew}

\begin{thebibliography}{57}%
\makeatletter
\providecommand \@ifxundefined [1]{%
 \@ifx{#1\undefined}
}%
\providecommand \@ifnum [1]{%
 \ifnum #1\expandafter \@firstoftwo
 \else \expandafter \@secondoftwo
 \fi
}%
\providecommand \@ifx [1]{%
 \ifx #1\expandafter \@firstoftwo
 \else \expandafter \@secondoftwo
 \fi
}%
\providecommand \natexlab [1]{#1}%
\providecommand \enquote  [1]{``#1''}%
\providecommand \bibnamefont  [1]{#1}%
\providecommand \bibfnamefont [1]{#1}%
\providecommand \citenamefont [1]{#1}%
\providecommand \href@noop [0]{\@secondoftwo}%
\providecommand \href [0]{\begingroup \@sanitize@url \@href}%
\providecommand \@href[1]{\@@startlink{#1}\@@href}%
\providecommand \@@href[1]{\endgroup#1\@@endlink}%
\providecommand \@sanitize@url [0]{\catcode `\\12\catcode `\$12\catcode
  `\&12\catcode `\#12\catcode `\^12\catcode `\_12\catcode `\%12\relax}%
\providecommand \@@startlink[1]{}%
\providecommand \@@endlink[0]{}%
\providecommand \url  [0]{\begingroup\@sanitize@url \@url }%
\providecommand \@url [1]{\endgroup\@href {#1}{\urlprefix }}%
\providecommand \urlprefix  [0]{URL }%
\providecommand \Eprint [0]{\href }%
\providecommand \doibase [0]{http://dx.doi.org/}%
\providecommand \selectlanguage [0]{\@gobble}%
\providecommand \bibinfo  [0]{\@secondoftwo}%
\providecommand \bibfield  [0]{\@secondoftwo}%
\providecommand \translation [1]{[#1]}%
\providecommand \BibitemOpen [0]{}%
\providecommand \bibitemStop [0]{}%
\providecommand \bibitemNoStop [0]{.\EOS\space}%
\providecommand \EOS [0]{\spacefactor3000\relax}%
\providecommand \BibitemShut  [1]{\csname bibitem#1\endcsname}%
\let\auto@bib@innerbib\@empty
\bibitem [{\citenamefont {Anderson}\ and\ \citenamefont
  {Wendt}(1995)}]{anderson1995computational}%
  \BibitemOpen
  \bibfield  {author} {\bibinfo {author} {\bibfnamefont {John~David}\
  \bibnamefont {Anderson}}\ and\ \bibinfo {author} {\bibfnamefont
  {J}~\bibnamefont {Wendt}},\ }\href@noop {} {\emph {\bibinfo {title}
  {Computational fluid dynamics}}}\ (\bibinfo  {publisher} {Springer},\
  \bibinfo {address} {New York},\ \bibinfo {year} {1995})\BibitemShut {NoStop}%
\bibitem [{\citenamefont {Hobbie}\ and\ \citenamefont
  {Roth}(2007)}]{hobbie2007intermediate}%
  \BibitemOpen
  \bibfield  {author} {\bibinfo {author} {\bibfnamefont {Russell~K}\
  \bibnamefont {Hobbie}}\ and\ \bibinfo {author} {\bibfnamefont {Bradley~J}\
  \bibnamefont {Roth}},\ }\href@noop {} {\emph {\bibinfo {title} {Intermediate
  physics for medicine and biology}}}\ (\bibinfo  {publisher} {Springer},\
  \bibinfo {address} {New York},\ \bibinfo {year} {2007})\BibitemShut {NoStop}%
\bibitem [{\citenamefont {Ghil}\ and\ \citenamefont
  {Childress}(2012)}]{ghil2012topics}%
  \BibitemOpen
  \bibfield  {author} {\bibinfo {author} {\bibfnamefont {Michael}\ \bibnamefont
  {Ghil}}\ and\ \bibinfo {author} {\bibfnamefont {Stephen}\ \bibnamefont
  {Childress}},\ }\href@noop {} {\emph {\bibinfo {title} {Topics in geophysical
  fluid dynamics: atmospheric dynamics, dynamo theory, and climate dynamics}}}\
  (\bibinfo  {publisher} {Springer},\ \bibinfo {address} {New York},\ \bibinfo
  {year} {2012})\BibitemShut {NoStop}%
\bibitem [{\citenamefont {Bishop}\ and\ \citenamefont
  {Johnson}(2011)}]{bishop2011mechanics}%
  \BibitemOpen
  \bibfield  {author} {\bibinfo {author} {\bibfnamefont {Richard
  Evelyn~Donohue}\ \bibnamefont {Bishop}}\ and\ \bibinfo {author}
  {\bibfnamefont {Daniel~Cowan}\ \bibnamefont {Johnson}},\ }\href@noop {}
  {\emph {\bibinfo {title} {The mechanics of vibration}}}\ (\bibinfo
  {publisher} {Cambridge University Press},\ \bibinfo {address} {Cambridge},\
  \bibinfo {year} {2011})\BibitemShut {NoStop}%
\bibitem [{\citenamefont {Wilcox}\ \emph {et~al.}(2006)\citenamefont {Wilcox}
  \emph {et~al.}}]{wilcox1998turbulence}%
  \BibitemOpen
  \bibfield  {author} {\bibinfo {author} {\bibfnamefont {David~C}\ \bibnamefont
  {Wilcox}} \emph {et~al.},\ }\href@noop {} {\emph {\bibinfo {title}
  {Turbulence modeling for CFD}}}\ (\bibinfo  {publisher} {DCW industries La
  Canada, CA},\ \bibinfo {year} {2006})\BibitemShut {NoStop}%
\bibitem [{\citenamefont {Schuster}(1984)}]{alligood1996chaos}%
  \BibitemOpen
  \bibfield  {author} {\bibinfo {author} {\bibfnamefont {H.~G.}\ \bibnamefont
  {Schuster}},\ }\href@noop {} {\emph {\bibinfo {title} {Deterministic Chaos.
  An Introduction}}}\ (\bibinfo  {publisher} {PhysikVerlag},\ \bibinfo
  {address} {Weinheim},\ \bibinfo {year} {1984})\BibitemShut {NoStop}%
\bibitem [{\citenamefont {Vicsek}(1992)}]{s1992fractal}%
  \BibitemOpen
  \bibfield  {author} {\bibinfo {author} {\bibfnamefont {Tam}\ \bibnamefont
  {Vicsek}},\ }\href@noop {} {\emph {\bibinfo {title} {Fractal growth
  phenomena}}}\ (\bibinfo  {publisher} {World scientific},\ \bibinfo {address}
  {Singapore},\ \bibinfo {year} {1992})\BibitemShut {NoStop}%
\bibitem [{\citenamefont {Dennis~Jr}\ and\ \citenamefont
  {Schnabel}(1996)}]{dennis1996numerical}%
  \BibitemOpen
  \bibfield  {author} {\bibinfo {author} {\bibfnamefont {John~E}\ \bibnamefont
  {Dennis~Jr}}\ and\ \bibinfo {author} {\bibfnamefont {Robert~B}\ \bibnamefont
  {Schnabel}},\ }\href@noop {} {\emph {\bibinfo {title} {Numerical methods for
  unconstrained optimization and nonlinear equations}}}\ (\bibinfo  {publisher}
  {Society for Industrial and Applied Mathematics},\ \bibinfo {address}
  {Philadelphia},\ \bibinfo {year} {1996})\BibitemShut {NoStop}%
\bibitem [{\citenamefont {Shor}(1999)}]{shor1999polynomial}%
  \BibitemOpen
  \bibfield  {author} {\bibinfo {author} {\bibfnamefont {Peter~W}\ \bibnamefont
  {Shor}},\ }\bibfield  {title} {\enquote {\bibinfo {title} {Polynomial-time
  algorithms for prime factorization and discrete logarithms on a quantum
  computer},}\ }\href@noop {} {\bibfield  {journal} {\bibinfo  {journal} {SIAM
  review}\ }\textbf {\bibinfo {volume} {41}},\ \bibinfo {pages} {303--332}
  (\bibinfo {year} {1999})}\BibitemShut {NoStop}%
\bibitem [{\citenamefont {Grover}(1996)}]{grover1996fast}%
  \BibitemOpen
  \bibfield  {author} {\bibinfo {author} {\bibfnamefont {Lov~K}\ \bibnamefont
  {Grover}},\ }\bibfield  {title} {\enquote {\bibinfo {title} {A fast quantum
  mechanical algorithm for database search},}\ }in\ \href@noop {} {\emph
  {\bibinfo {booktitle} {Proceedings of the twenty-eighth annual ACM symposium
  on Theory of computing}}}\ (\bibinfo  {publisher} {ACM},\ \bibinfo {address}
  {New York},\ \bibinfo {year} {1996})\ pp.\ \bibinfo {pages}
  {212--219}\BibitemShut {NoStop}%
\bibitem [{\citenamefont {Harrow}\ \emph {et~al.}(2009)\citenamefont {Harrow},
  \citenamefont {Hassidim},\ and\ \citenamefont {Lloyd}}]{harrow2009quantum}%
  \BibitemOpen
  \bibfield  {author} {\bibinfo {author} {\bibfnamefont {Aram~W.}\ \bibnamefont
  {Harrow}}, \bibinfo {author} {\bibfnamefont {Avinatan}\ \bibnamefont
  {Hassidim}}, \ and\ \bibinfo {author} {\bibfnamefont {Seth}\ \bibnamefont
  {Lloyd}},\ }\bibfield  {title} {\enquote {\bibinfo {title} {Quantum
  {Algorithm} for {Linear} {Systems} of {Equations}},}\ }\href {\doibase
  10.1103/PhysRevLett.103.150502} {\bibfield  {journal} {\bibinfo  {journal}
  {Physical Review Letters}\ }\textbf {\bibinfo {volume} {103}},\ \bibinfo
  {pages} {150502} (\bibinfo {year} {2009})}\BibitemShut {NoStop}%
\bibitem [{\citenamefont {Childs}\ \emph {et~al.}(2017)\citenamefont {Childs},
  \citenamefont {Kothari},\ and\ \citenamefont {Somma}}]{childs2017quantum}%
  \BibitemOpen
  \bibfield  {author} {\bibinfo {author} {\bibfnamefont {Andrew~M.}\
  \bibnamefont {Childs}}, \bibinfo {author} {\bibfnamefont {Robin}\
  \bibnamefont {Kothari}}, \ and\ \bibinfo {author} {\bibfnamefont
  {Rolando~D.}\ \bibnamefont {Somma}},\ }\bibfield  {title} {\enquote {\bibinfo
  {title} {Quantum {Algorithm} for {Systems} of {Linear} {Equations} with
  {Exponentially} {Improved} {Dependence} on {Precision}},}\ }\href {\doibase
  10.1137/16M1087072} {\bibfield  {journal} {\bibinfo  {journal} {SIAM Journal
  on Computing}\ }\textbf {\bibinfo {volume} {46}},\ \bibinfo {pages}
  {1920--1950} (\bibinfo {year} {2017})}\BibitemShut {NoStop}%
\bibitem [{\citenamefont {Suba\ifmmode \mbox{\c{s}}\else \c{s}\fi{}\ifmmode
  \imath \else~\i \fi{}}\ \emph {et~al.}(2019)\citenamefont {Suba\ifmmode
  \mbox{\c{s}}\else \c{s}\fi{}\ifmmode \imath \else~\i \fi{}}, \citenamefont
  {Somma},\ and\ \citenamefont {Orsucci}}]{subacsi2019quantum}%
  \BibitemOpen
  \bibfield  {author} {\bibinfo {author} {\bibfnamefont {Yi\ifmmode
  \breve{g}\else~\u{g}\fi{}it}\ \bibnamefont {Suba\ifmmode \mbox{\c{s}}\else
  \c{s}\fi{}\ifmmode \imath \else~\i \fi{}}}, \bibinfo {author} {\bibfnamefont
  {Rolando~D.}\ \bibnamefont {Somma}}, \ and\ \bibinfo {author} {\bibfnamefont
  {Davide}\ \bibnamefont {Orsucci}},\ }\bibfield  {title} {\enquote {\bibinfo
  {title} {Quantum algorithms for systems of linear equations inspired by
  adiabatic quantum computing},}\ }\href {\doibase
  10.1103/PhysRevLett.122.060504} {\bibfield  {journal} {\bibinfo  {journal}
  {Phys. Rev. Lett.}\ }\textbf {\bibinfo {volume} {122}},\ \bibinfo {pages}
  {060504} (\bibinfo {year} {2019})}\BibitemShut {NoStop}%
\bibitem [{\citenamefont {Xu}\ \emph {et~al.}(2021)\citenamefont {Xu},
  \citenamefont {Sun}, \citenamefont {Endo}, \citenamefont {Li}, \citenamefont
  {Benjamin},\ and\ \citenamefont {Yuan}}]{xu2021variational}%
  \BibitemOpen
  \bibfield  {author} {\bibinfo {author} {\bibfnamefont {Xiaosi}\ \bibnamefont
  {Xu}}, \bibinfo {author} {\bibfnamefont {Jinzhao}\ \bibnamefont {Sun}},
  \bibinfo {author} {\bibfnamefont {Suguru}\ \bibnamefont {Endo}}, \bibinfo
  {author} {\bibfnamefont {Ying}\ \bibnamefont {Li}}, \bibinfo {author}
  {\bibfnamefont {Simon~C.}\ \bibnamefont {Benjamin}}, \ and\ \bibinfo {author}
  {\bibfnamefont {Xiao}\ \bibnamefont {Yuan}},\ }\bibfield  {title} {\enquote
  {\bibinfo {title} {Variational algorithms for linear algebra},}\ }\href
  {\doibase https://doi.org/10.1016/j.scib.2021.06.023} {\bibfield  {journal}
  {\bibinfo  {journal} {Science Bulletin}\ }\textbf {\bibinfo {volume} {66}},\
  \bibinfo {pages} {2181--2188} (\bibinfo {year} {2021})}\BibitemShut {NoStop}%
\bibitem [{\citenamefont {Clader}\ \emph {et~al.}(2013)\citenamefont {Clader},
  \citenamefont {Jacobs},\ and\ \citenamefont
  {Sprouse}}]{clader2013preconditioned}%
  \BibitemOpen
  \bibfield  {author} {\bibinfo {author} {\bibfnamefont {B~David}\ \bibnamefont
  {Clader}}, \bibinfo {author} {\bibfnamefont {Bryan~C}\ \bibnamefont
  {Jacobs}}, \ and\ \bibinfo {author} {\bibfnamefont {Chad~R}\ \bibnamefont
  {Sprouse}},\ }\bibfield  {title} {\enquote {\bibinfo {title} {Preconditioned
  quantum linear system algorithm},}\ }\href
  {http://dx.doi.org/10.1103/PhysRevLett.110.250504} {\bibfield  {journal}
  {\bibinfo  {journal} {Physical review letters}\ }\textbf {\bibinfo {volume}
  {110}},\ \bibinfo {pages} {250504} (\bibinfo {year} {2013})}\BibitemShut
  {NoStop}%
\bibitem [{\citenamefont {Berry}(2014)}]{berry2014highorder}%
  \BibitemOpen
  \bibfield  {author} {\bibinfo {author} {\bibfnamefont {Dominic~W}\
  \bibnamefont {Berry}},\ }\bibfield  {title} {\enquote {\bibinfo {title}
  {High-order quantum algorithm for solving linear differential equations},}\
  }\href {\doibase 10.1088/1751-8113/47/10/105301} {\bibfield  {journal}
  {\bibinfo  {journal} {Journal of Physics A}\ }\textbf {\bibinfo {volume}
  {47}},\ \bibinfo {pages} {105301} (\bibinfo {year} {2014})}\BibitemShut
  {NoStop}%
\bibitem [{\citenamefont {Montanaro}\ and\ \citenamefont
  {Pallister}(2016)}]{Montanaro_2016}%
  \BibitemOpen
  \bibfield  {author} {\bibinfo {author} {\bibfnamefont {Ashley}\ \bibnamefont
  {Montanaro}}\ and\ \bibinfo {author} {\bibfnamefont {Sam}\ \bibnamefont
  {Pallister}},\ }\bibfield  {title} {\enquote {\bibinfo {title} {Quantum
  algorithms and the finite element method},}\ }\href
  {http://dx.doi.org/10.1103/PhysRevA.93.032324} {\bibfield  {journal}
  {\bibinfo  {journal} {Physical Review A}\ }\textbf {\bibinfo {volume} {93}},\
  \bibinfo {pages} {032324} (\bibinfo {year} {2016})}\BibitemShut {NoStop}%
\bibitem [{\citenamefont {Berry}\ \emph {et~al.}(2017)\citenamefont {Berry},
  \citenamefont {Childs}, \citenamefont {Ostrander},\ and\ \citenamefont
  {Wang}}]{berry2017quantum}%
  \BibitemOpen
  \bibfield  {author} {\bibinfo {author} {\bibfnamefont {Dominic~W.}\
  \bibnamefont {Berry}}, \bibinfo {author} {\bibfnamefont {Andrew~M.}\
  \bibnamefont {Childs}}, \bibinfo {author} {\bibfnamefont {Aaron}\
  \bibnamefont {Ostrander}}, \ and\ \bibinfo {author} {\bibfnamefont {Guoming}\
  \bibnamefont {Wang}},\ }\bibfield  {title} {\enquote {\bibinfo {title}
  {Quantum algorithm for linear differential equations with exponentially
  improved dependence on precision},}\ }\href {\doibase
  10.1007/s00220-017-3002-y} {\bibfield  {journal} {\bibinfo  {journal}
  {Communications in Mathematical Physics}\ }\textbf {\bibinfo {volume}
  {356}},\ \bibinfo {pages} {1057–1081} (\bibinfo {year} {2017})}\BibitemShut
  {NoStop}%
\bibitem [{\citenamefont {Xin}\ \emph {et~al.}(2020)\citenamefont {Xin},
  \citenamefont {Wei}, \citenamefont {Cui}, \citenamefont {Xiao}, \citenamefont
  {Arrazola}, \citenamefont {Lamata}, \citenamefont {Kong}, \citenamefont {Lu},
  \citenamefont {Solano},\ and\ \citenamefont {Long}}]{xin2020quantum}%
  \BibitemOpen
  \bibfield  {author} {\bibinfo {author} {\bibfnamefont {Tao}\ \bibnamefont
  {Xin}}, \bibinfo {author} {\bibfnamefont {Shijie}\ \bibnamefont {Wei}},
  \bibinfo {author} {\bibfnamefont {Jianlian}\ \bibnamefont {Cui}}, \bibinfo
  {author} {\bibfnamefont {Junxiang}\ \bibnamefont {Xiao}}, \bibinfo {author}
  {\bibfnamefont {I{\~n}igo}\ \bibnamefont {Arrazola}}, \bibinfo {author}
  {\bibfnamefont {Lucas}\ \bibnamefont {Lamata}}, \bibinfo {author}
  {\bibfnamefont {Xiangyu}\ \bibnamefont {Kong}}, \bibinfo {author}
  {\bibfnamefont {Dawei}\ \bibnamefont {Lu}}, \bibinfo {author} {\bibfnamefont
  {Enrique}\ \bibnamefont {Solano}}, \ and\ \bibinfo {author} {\bibfnamefont
  {Guilu}\ \bibnamefont {Long}},\ }\bibfield  {title} {\enquote {\bibinfo
  {title} {Quantum algorithm for solving linear differential equations: Theory
  and experiment},}\ }\href {http://dx.doi.org/10.1103/PhysRevA.101.032307}
  {\bibfield  {journal} {\bibinfo  {journal} {Physical Review A}\ }\textbf
  {\bibinfo {volume} {101}},\ \bibinfo {pages} {032307} (\bibinfo {year}
  {2020})}\BibitemShut {NoStop}%
\bibitem [{\citenamefont {Cao}\ \emph {et~al.}(2013)\citenamefont {Cao},
  \citenamefont {Papageorgiou}, \citenamefont {Petras}, \citenamefont {Traub},\
  and\ \citenamefont {Kais}}]{cao2013quantum}%
  \BibitemOpen
  \bibfield  {author} {\bibinfo {author} {\bibfnamefont {Yudong}\ \bibnamefont
  {Cao}}, \bibinfo {author} {\bibfnamefont {Anargyros}\ \bibnamefont
  {Papageorgiou}}, \bibinfo {author} {\bibfnamefont {Iasonas}\ \bibnamefont
  {Petras}}, \bibinfo {author} {\bibfnamefont {Joseph}\ \bibnamefont {Traub}},
  \ and\ \bibinfo {author} {\bibfnamefont {Sabre}\ \bibnamefont {Kais}},\
  }\bibfield  {title} {\enquote {\bibinfo {title} {Quantum algorithm and
  circuit design solving the poisson equation},}\ }\href {\doibase
  10.1088/1367-2630/15/1/013021} {\bibfield  {journal} {\bibinfo  {journal}
  {New Journal of Physics}\ }\textbf {\bibinfo {volume} {15}},\ \bibinfo
  {pages} {013021} (\bibinfo {year} {2013})}\BibitemShut {NoStop}%
\bibitem [{\citenamefont {Costa}\ \emph {et~al.}(2019)\citenamefont {Costa},
  \citenamefont {Jordan},\ and\ \citenamefont {Ostrander}}]{costa2019quantum}%
  \BibitemOpen
  \bibfield  {author} {\bibinfo {author} {\bibfnamefont {Pedro~CS}\
  \bibnamefont {Costa}}, \bibinfo {author} {\bibfnamefont {Stephen}\
  \bibnamefont {Jordan}}, \ and\ \bibinfo {author} {\bibfnamefont {Aaron}\
  \bibnamefont {Ostrander}},\ }\bibfield  {title} {\enquote {\bibinfo {title}
  {Quantum algorithm for simulating the wave equation},}\ }\href
  {http://dx.doi.org/10.1103/PhysRevA.99.012323} {\bibfield  {journal}
  {\bibinfo  {journal} {Physical Review A}\ }\textbf {\bibinfo {volume} {99}},\
  \bibinfo {pages} {012323} (\bibinfo {year} {2019})}\BibitemShut {NoStop}%
\bibitem [{\citenamefont {Fillion-Gourdeau}\ \emph {et~al.}(2017)\citenamefont
  {Fillion-Gourdeau}, \citenamefont {MacLean},\ and\ \citenamefont
  {Laflamme}}]{fillion2017quantum}%
  \BibitemOpen
  \bibfield  {author} {\bibinfo {author} {\bibfnamefont {Fran{\c{c}}ois}\
  \bibnamefont {Fillion-Gourdeau}}, \bibinfo {author} {\bibfnamefont {Steve}\
  \bibnamefont {MacLean}}, \ and\ \bibinfo {author} {\bibfnamefont {Raymond}\
  \bibnamefont {Laflamme}},\ }\bibfield  {title} {\enquote {\bibinfo {title}
  {Algorithm for the solution of the {Dirac} equation on digital quantum
  computers},}\ }\href {http://dx.doi.org/10.1103/PhysRevA.95.042343}
  {\bibfield  {journal} {\bibinfo  {journal} {Physical Review A}\ }\textbf
  {\bibinfo {volume} {95}},\ \bibinfo {pages} {042343} (\bibinfo {year}
  {2017})}\BibitemShut {NoStop}%
\bibitem [{\citenamefont {Engel}\ \emph {et~al.}(2019)\citenamefont {Engel},
  \citenamefont {Smith},\ and\ \citenamefont {Parker}}]{engel2019quantum}%
  \BibitemOpen
  \bibfield  {author} {\bibinfo {author} {\bibfnamefont {Alexander}\
  \bibnamefont {Engel}}, \bibinfo {author} {\bibfnamefont {Graeme}\
  \bibnamefont {Smith}}, \ and\ \bibinfo {author} {\bibfnamefont {Scott~E}\
  \bibnamefont {Parker}},\ }\bibfield  {title} {\enquote {\bibinfo {title}
  {Quantum algorithm for the {Vlasov} equation},}\ }\href
  {http://dx.doi.org/10.1103/PhysRevA.100.062315} {\bibfield  {journal}
  {\bibinfo  {journal} {Physical Review A}\ }\textbf {\bibinfo {volume}
  {100}},\ \bibinfo {pages} {062315} (\bibinfo {year} {2019})}\BibitemShut
  {NoStop}%
\bibitem [{\citenamefont {Arrazola}\ \emph {et~al.}(2019)\citenamefont
  {Arrazola}, \citenamefont {Kalajdzievski}, \citenamefont {Weedbrook},\ and\
  \citenamefont {Lloyd}}]{arrazola2019quantum}%
  \BibitemOpen
  \bibfield  {author} {\bibinfo {author} {\bibfnamefont {Juan~Miguel}\
  \bibnamefont {Arrazola}}, \bibinfo {author} {\bibfnamefont {Timjan}\
  \bibnamefont {Kalajdzievski}}, \bibinfo {author} {\bibfnamefont {Christian}\
  \bibnamefont {Weedbrook}}, \ and\ \bibinfo {author} {\bibfnamefont {Seth}\
  \bibnamefont {Lloyd}},\ }\bibfield  {title} {\enquote {\bibinfo {title}
  {Quantum algorithm for nonhomogeneous linear partial differential
  equations},}\ }\href {http://dx.doi.org/10.1103/PhysRevA.100.032306}
  {\bibfield  {journal} {\bibinfo  {journal} {Physical Review A}\ }\textbf
  {\bibinfo {volume} {100}},\ \bibinfo {pages} {032306} (\bibinfo {year}
  {2019})}\BibitemShut {NoStop}%
\bibitem [{\citenamefont {Linden}\ \emph {et~al.}(2022)\citenamefont {Linden},
  \citenamefont {Montanaro},\ and\ \citenamefont {Shao}}]{linden2022quantum}%
  \BibitemOpen
  \bibfield  {author} {\bibinfo {author} {\bibfnamefont {Noah}\ \bibnamefont
  {Linden}}, \bibinfo {author} {\bibfnamefont {Ashley}\ \bibnamefont
  {Montanaro}}, \ and\ \bibinfo {author} {\bibfnamefont {Changpeng}\
  \bibnamefont {Shao}},\ }\bibfield  {title} {\enquote {\bibinfo {title}
  {Quantum vs. classical algorithms for solving the heat equation},}\ }\href
  {\doibase 10.1007/s00220-022-04442-6} {\bibfield  {journal} {\bibinfo
  {journal} {Communications in Mathematical Physics}\ } (\bibinfo {year}
  {2022}),\ 10.1007/s00220-022-04442-6}\BibitemShut {NoStop}%
\bibitem [{\citenamefont {Childs}\ and\ \citenamefont
  {Liu}(2020)}]{childs2020quantumspectral}%
  \BibitemOpen
  \bibfield  {author} {\bibinfo {author} {\bibfnamefont {Andrew~M.}\
  \bibnamefont {Childs}}\ and\ \bibinfo {author} {\bibfnamefont {Jin-Peng}\
  \bibnamefont {Liu}},\ }\bibfield  {title} {\enquote {\bibinfo {title}
  {Quantum spectral methods for differential equations},}\ }\href {\doibase
  10.1007/s00220-020-03699-z} {\bibfield  {journal} {\bibinfo  {journal}
  {Communications in Mathematical Physics}\ }\textbf {\bibinfo {volume}
  {375}},\ \bibinfo {pages} {1427–1457} (\bibinfo {year} {2020})}\BibitemShut
  {NoStop}%
\bibitem [{\citenamefont {Childs}\ \emph {et~al.}(2021)\citenamefont {Childs},
  \citenamefont {Liu},\ and\ \citenamefont
  {Ostrander}}]{childs2020highprecision}%
  \BibitemOpen
  \bibfield  {author} {\bibinfo {author} {\bibfnamefont {Andrew~M.}\
  \bibnamefont {Childs}}, \bibinfo {author} {\bibfnamefont {Jin-Peng}\
  \bibnamefont {Liu}}, \ and\ \bibinfo {author} {\bibfnamefont {Aaron}\
  \bibnamefont {Ostrander}},\ }\bibfield  {title} {\enquote {\bibinfo {title}
  {High-precision quantum algorithms for partial differential equations},}\
  }\href {\doibase 10.22331/q-2021-11-10-574} {\bibfield  {journal} {\bibinfo
  {journal} {Quantum}\ }\textbf {\bibinfo {volume} {5}},\ \bibinfo {pages}
  {574} (\bibinfo {year} {2021})}\BibitemShut {NoStop}%
\bibitem [{\citenamefont {Nielsen}\ and\ \citenamefont
  {Chuang}(2002)}]{nielsen_quantum_2002}%
  \BibitemOpen
  \bibfield  {author} {\bibinfo {author} {\bibfnamefont {Michael~A.}\
  \bibnamefont {Nielsen}}\ and\ \bibinfo {author} {\bibfnamefont {Isaac}\
  \bibnamefont {Chuang}},\ }\bibfield  {title} {\enquote {\bibinfo {title}
  {Quantum computation and quantum information},}\ }\href {\doibase
  10.1119/1.1463744} {\bibfield  {journal} {\bibinfo  {journal} {American
  Journal of Physics}\ }\textbf {\bibinfo {volume} {70}},\ \bibinfo {pages}
  {558--559} (\bibinfo {year} {2002})}\BibitemShut {NoStop}%
\bibitem [{\citenamefont {Leyton}\ and\ \citenamefont
  {Osborne}(2008)}]{leyton2008quantum}%
  \BibitemOpen
  \bibfield  {author} {\bibinfo {author} {\bibfnamefont {Sarah~K.}\
  \bibnamefont {Leyton}}\ and\ \bibinfo {author} {\bibfnamefont {Tobias~J.}\
  \bibnamefont {Osborne}},\ }\href@noop {} {\enquote {\bibinfo {title} {A
  quantum algorithm to solve nonlinear differential equations},}\ } (\bibinfo
  {year} {2008}),\ \Eprint {http://arxiv.org/abs/0812.4423} {arXiv:0812.4423
  [quant-ph]} \BibitemShut {NoStop}%
\bibitem [{\citenamefont {Lloyd}\ \emph {et~al.}()\citenamefont {Lloyd},
  \citenamefont {Palma}, \citenamefont {Gokler}, \citenamefont {Kiani},
  \citenamefont {Liu}, \citenamefont {Marvian}, \citenamefont {Tennie},\ and\
  \citenamefont {Palmer}}]{lloyd2020quantum}%
  \BibitemOpen
  \bibfield  {author} {\bibinfo {author} {\bibfnamefont {Seth}\ \bibnamefont
  {Lloyd}}, \bibinfo {author} {\bibfnamefont {Giacomo~De}\ \bibnamefont
  {Palma}}, \bibinfo {author} {\bibfnamefont {Can}\ \bibnamefont {Gokler}},
  \bibinfo {author} {\bibfnamefont {Bobak}\ \bibnamefont {Kiani}}, \bibinfo
  {author} {\bibfnamefont {Zi-Wen}\ \bibnamefont {Liu}}, \bibinfo {author}
  {\bibfnamefont {Milad}\ \bibnamefont {Marvian}}, \bibinfo {author}
  {\bibfnamefont {Felix}\ \bibnamefont {Tennie}}, \ and\ \bibinfo {author}
  {\bibfnamefont {Tim}\ \bibnamefont {Palmer}},\ }\href@noop {} {\enquote
  {\bibinfo {title} {Quantum algorithm for nonlinear differential equations},}\
  }\Eprint {http://arxiv.org/abs/2011.06571} {arXiv:2011.06571 [quant-ph]}
  \BibitemShut {NoStop}%
\bibitem [{\citenamefont {Liu}\ \emph {et~al.}(2021)\citenamefont {Liu},
  \citenamefont {Kolden}, \citenamefont {Krovi}, \citenamefont {Loureiro},
  \citenamefont {Trivisa},\ and\ \citenamefont {Childs}}]{liu2020efficient}%
  \BibitemOpen
  \bibfield  {author} {\bibinfo {author} {\bibfnamefont {Jin-Peng}\
  \bibnamefont {Liu}}, \bibinfo {author} {\bibfnamefont {Herman~{\O}ie}\
  \bibnamefont {Kolden}}, \bibinfo {author} {\bibfnamefont {Hari~K.}\
  \bibnamefont {Krovi}}, \bibinfo {author} {\bibfnamefont {Nuno~F.}\
  \bibnamefont {Loureiro}}, \bibinfo {author} {\bibfnamefont {Konstantina}\
  \bibnamefont {Trivisa}}, \ and\ \bibinfo {author} {\bibfnamefont {Andrew~M.}\
  \bibnamefont {Childs}},\ }\bibfield  {title} {\enquote {\bibinfo {title}
  {Efficient quantum algorithm for dissipative nonlinear differential
  equations},}\ }\href {http://dx.doi.org/10.1073/pnas.2026805118} {\bibfield
  {journal} {\bibinfo  {journal} {Proceedings of the National Academy of
  Sciences}\ }\textbf {\bibinfo {volume} {118}},\ \bibinfo {pages}
  {e2026805118} (\bibinfo {year} {2021})}\BibitemShut {NoStop}%
\bibitem [{\citenamefont {Kyriienko}\ \emph {et~al.}(2021)\citenamefont
  {Kyriienko}, \citenamefont {Paine},\ and\ \citenamefont
  {Elfving}}]{kyriienko2021solving}%
  \BibitemOpen
  \bibfield  {author} {\bibinfo {author} {\bibfnamefont {Oleksandr}\
  \bibnamefont {Kyriienko}}, \bibinfo {author} {\bibfnamefont {Annie~E.}\
  \bibnamefont {Paine}}, \ and\ \bibinfo {author} {\bibfnamefont {Vincent~E.}\
  \bibnamefont {Elfving}},\ }\bibfield  {title} {\enquote {\bibinfo {title}
  {Solving nonlinear differential equations with differentiable quantum
  circuits},}\ }\href {\doibase 10.1103/PhysRevA.103.052416} {\bibfield
  {journal} {\bibinfo  {journal} {Phys. Rev. A}\ }\textbf {\bibinfo {volume}
  {103}},\ \bibinfo {pages} {052416} (\bibinfo {year} {2021})}\BibitemShut
  {NoStop}%
\bibitem [{\citenamefont {Xue}\ \emph {et~al.}(2021{\natexlab{a}})\citenamefont
  {Xue}, \citenamefont {Wu},\ and\ \citenamefont
  {Guo}}]{xue2021quantumHomotopy}%
  \BibitemOpen
  \bibfield  {author} {\bibinfo {author} {\bibfnamefont {Cheng}\ \bibnamefont
  {Xue}}, \bibinfo {author} {\bibfnamefont {Yu-Chun}\ \bibnamefont {Wu}}, \
  and\ \bibinfo {author} {\bibfnamefont {Guo-Ping}\ \bibnamefont {Guo}},\
  }\bibfield  {title} {\enquote {\bibinfo {title} {Quantum homotopy
  perturbation method for nonlinear dissipative ordinary differential
  equations},}\ }\href {\doibase 10.1088/1367-2630/ac3eff} {\bibfield
  {journal} {\bibinfo  {journal} {New Journal of Physics}\ }\textbf {\bibinfo
  {volume} {23}},\ \bibinfo {pages} {123035} (\bibinfo {year}
  {2021}{\natexlab{a}})}\BibitemShut {NoStop}%
\bibitem [{\citenamefont {Krovi}()}]{krovi2022improved}%
  \BibitemOpen
  \bibfield  {author} {\bibinfo {author} {\bibfnamefont {Hari}\ \bibnamefont
  {Krovi}},\ }\bibfield  {title} {\enquote {\bibinfo {title} {Improved quantum
  algorithms for linear and nonlinear differential equations},}\ }\href@noop {}
  {\ }\Eprint {http://arxiv.org/abs/arXiv:2202.01054} {arXiv:2202.01054}
  \BibitemShut {NoStop}%
\bibitem [{\citenamefont {Lin}\ \emph {et~al.}()\citenamefont {Lin},
  \citenamefont {Lowrie}, \citenamefont {Aslangil}, \citenamefont
  {Suba{\c{s}}{\i}},\ and\ \citenamefont {Sornborger}}]{lin2022koopman}%
  \BibitemOpen
  \bibfield  {author} {\bibinfo {author} {\bibfnamefont {Yen~Ting}\
  \bibnamefont {Lin}}, \bibinfo {author} {\bibfnamefont {Robert~B}\
  \bibnamefont {Lowrie}}, \bibinfo {author} {\bibfnamefont {Denis}\
  \bibnamefont {Aslangil}}, \bibinfo {author} {\bibfnamefont {Yi{\u{g}}it}\
  \bibnamefont {Suba{\c{s}}{\i}}}, \ and\ \bibinfo {author} {\bibfnamefont
  {Andrew~T}\ \bibnamefont {Sornborger}},\ }\bibfield  {title} {\enquote
  {\bibinfo {title} {Koopman von {Neumann} mechanics and the {Koopman}
  representation: A perspective on solving nonlinear dynamical systems with
  quantum computers},}\ }\href@noop {} {\ }\Eprint
  {http://arxiv.org/abs/arXiv:2202.02188} {arXiv:2202.02188} \BibitemShut
  {NoStop}%
\bibitem [{\citenamefont {Jin}\ and\ \citenamefont {Liu}()}]{jin2022quantum}%
  \BibitemOpen
  \bibfield  {author} {\bibinfo {author} {\bibfnamefont {Shi}\ \bibnamefont
  {Jin}}\ and\ \bibinfo {author} {\bibfnamefont {Nana}\ \bibnamefont {Liu}},\
  }\bibfield  {title} {\enquote {\bibinfo {title} {Quantum algorithms for
  computing observables of nonlinear partial differential equations},}\
  }\href@noop {} {\ }\Eprint {http://arxiv.org/abs/arXiv:2202.07834}
  {arXiv:2202.07834} \BibitemShut {NoStop}%
\bibitem [{\citenamefont {Lubasch}\ \emph {et~al.}(2020)\citenamefont
  {Lubasch}, \citenamefont {Joo}, \citenamefont {Moinier}, \citenamefont
  {Kiffner},\ and\ \citenamefont {Jaksch}}]{lubasch2020variational}%
  \BibitemOpen
  \bibfield  {author} {\bibinfo {author} {\bibfnamefont {Michael}\ \bibnamefont
  {Lubasch}}, \bibinfo {author} {\bibfnamefont {Jaewoo}\ \bibnamefont {Joo}},
  \bibinfo {author} {\bibfnamefont {Pierre}\ \bibnamefont {Moinier}}, \bibinfo
  {author} {\bibfnamefont {Martin}\ \bibnamefont {Kiffner}}, \ and\ \bibinfo
  {author} {\bibfnamefont {Dieter}\ \bibnamefont {Jaksch}},\ }\bibfield
  {title} {\enquote {\bibinfo {title} {Variational quantum algorithms for
  nonlinear problems},}\ }\href {\doibase 10.1103/PhysRevA.101.010301}
  {\bibfield  {journal} {\bibinfo  {journal} {Physical Review A}\ }\textbf
  {\bibinfo {volume} {101}},\ \bibinfo {pages} {010301} (\bibinfo {year}
  {2020})}\BibitemShut {NoStop}%
\bibitem [{\citenamefont {Chen}\ \emph {et~al.}(2022)\citenamefont {Chen},
  \citenamefont {Xue}, \citenamefont {Chen}, \citenamefont {Lu}, \citenamefont
  {Wu}, \citenamefont {Ding}, \citenamefont {Huang},\ and\ \citenamefont
  {Guo}}]{chen2022quantum}%
  \BibitemOpen
  \bibfield  {author} {\bibinfo {author} {\bibfnamefont {Zhao-Yun}\
  \bibnamefont {Chen}}, \bibinfo {author} {\bibfnamefont {Cheng}\ \bibnamefont
  {Xue}}, \bibinfo {author} {\bibfnamefont {Si-Ming}\ \bibnamefont {Chen}},
  \bibinfo {author} {\bibfnamefont {Bing-Han}\ \bibnamefont {Lu}}, \bibinfo
  {author} {\bibfnamefont {Yu-Chun}\ \bibnamefont {Wu}}, \bibinfo {author}
  {\bibfnamefont {Ju-Chun}\ \bibnamefont {Ding}}, \bibinfo {author}
  {\bibfnamefont {Sheng-Hong}\ \bibnamefont {Huang}}, \ and\ \bibinfo {author}
  {\bibfnamefont {Guo-Ping}\ \bibnamefont {Guo}},\ }\bibfield  {title}
  {\enquote {\bibinfo {title} {Quantum approach to accelerate finite volume
  method on steady computational fluid dynamics problems},}\ }\href {\doibase
  10.1007/s11128-022-03478-w} {\bibfield  {journal} {\bibinfo  {journal}
  {Quantum Inf. Process.}\ }\textbf {\bibinfo {volume} {21}},\ \bibinfo {pages}
  {137} (\bibinfo {year} {2022})}\BibitemShut {NoStop}%
\bibitem [{\citenamefont {Ljubomir}(2022)}]{budinski2021quantum}%
  \BibitemOpen
  \bibfield  {author} {\bibinfo {author} {\bibfnamefont {Budinski}\
  \bibnamefont {Ljubomir}},\ }\bibfield  {title} {\enquote {\bibinfo {title}
  {Quantum algorithm for the {Navier–Stokes} equations by using the
  streamfunction-vorticity formulation and the lattice {Boltzmann} method},}\
  }\href {\doibase 10.1142/S0219749921500398} {\bibfield  {journal} {\bibinfo
  {journal} {Int. J. Quantum. Inf.}\ }\textbf {\bibinfo {volume} {20}},\
  \bibinfo {pages} {2150039} (\bibinfo {year} {2022})}\BibitemShut {NoStop}%
\bibitem [{\citenamefont {Qian}\ \emph {et~al.}()\citenamefont {Qian},
  \citenamefont {Huang},\ and\ \citenamefont {Long}}]{qian2019quantum}%
  \BibitemOpen
  \bibfield  {author} {\bibinfo {author} {\bibfnamefont {Peng}\ \bibnamefont
  {Qian}}, \bibinfo {author} {\bibfnamefont {Wei-Cong}\ \bibnamefont {Huang}},
  \ and\ \bibinfo {author} {\bibfnamefont {Gui-Lu}\ \bibnamefont {Long}},\
  }\href@noop {} {\enquote {\bibinfo {title} {A quantum algorithm for solving
  systems of nonlinear algebraic equations},}\ }\Eprint
  {http://arxiv.org/abs/arXiv:1903.05608} {arXiv:1903.05608} \BibitemShut
  {NoStop}%
\bibitem [{\citenamefont {Xue}\ \emph {et~al.}(2021{\natexlab{b}})\citenamefont
  {Xue}, \citenamefont {Wu},\ and\ \citenamefont {Guo}}]{xue2021quantum}%
  \BibitemOpen
  \bibfield  {author} {\bibinfo {author} {\bibfnamefont {Cheng}\ \bibnamefont
  {Xue}}, \bibinfo {author} {\bibfnamefont {Yu-Chun}\ \bibnamefont {Wu}}, \
  and\ \bibinfo {author} {\bibfnamefont {Guo-Ping}\ \bibnamefont {Guo}},\
  }\bibfield  {title} {\enquote {\bibinfo {title} {Quantum {Newton's} method
  for solving the system of nonlinear equations},}\ }\href {\doibase
  10.1142/S201032472140004X} {\bibfield  {journal} {\bibinfo  {journal} {SPIN}\
  }\textbf {\bibinfo {volume} {11}},\ \bibinfo {pages} {2140004} (\bibinfo
  {year} {2021}{\natexlab{b}})}\BibitemShut {NoStop}%
\bibitem [{\citenamefont {Giovannetti}\ \emph
  {et~al.}(2008{\natexlab{a}})\citenamefont {Giovannetti}, \citenamefont
  {Lloyd},\ and\ \citenamefont {Maccone}}]{giovannetti2008quantum}%
  \BibitemOpen
  \bibfield  {author} {\bibinfo {author} {\bibfnamefont {Vittorio}\
  \bibnamefont {Giovannetti}}, \bibinfo {author} {\bibfnamefont {Seth}\
  \bibnamefont {Lloyd}}, \ and\ \bibinfo {author} {\bibfnamefont {Lorenzo}\
  \bibnamefont {Maccone}},\ }\bibfield  {title} {\enquote {\bibinfo {title}
  {Quantum random access memory},}\ }\href {\doibase
  10.1103/PhysRevLett.100.160501} {\bibfield  {journal} {\bibinfo  {journal}
  {Phys. Rev. Lett.}\ }\textbf {\bibinfo {volume} {100}},\ \bibinfo {pages}
  {160501} (\bibinfo {year} {2008}{\natexlab{a}})}\BibitemShut {NoStop}%
\bibitem [{\citenamefont {Giovannetti}\ \emph
  {et~al.}(2008{\natexlab{b}})\citenamefont {Giovannetti}, \citenamefont
  {Lloyd},\ and\ \citenamefont {Maccone}}]{giovannetti2008architectures}%
  \BibitemOpen
  \bibfield  {author} {\bibinfo {author} {\bibfnamefont {Vittorio}\
  \bibnamefont {Giovannetti}}, \bibinfo {author} {\bibfnamefont {Seth}\
  \bibnamefont {Lloyd}}, \ and\ \bibinfo {author} {\bibfnamefont {Lorenzo}\
  \bibnamefont {Maccone}},\ }\bibfield  {title} {\enquote {\bibinfo {title}
  {Architectures for a quantum random access memory},}\ }\href {\doibase
  10.1103/PhysRevA.78.052310} {\bibfield  {journal} {\bibinfo  {journal} {Phys.
  Rev. A}\ }\textbf {\bibinfo {volume} {78}},\ \bibinfo {pages} {052310}
  (\bibinfo {year} {2008}{\natexlab{b}})}\BibitemShut {NoStop}%
\bibitem [{\citenamefont {Kerenidis}\ and\ \citenamefont
  {Prakash}()}]{kerenidis2016quantum}%
  \BibitemOpen
  \bibfield  {author} {\bibinfo {author} {\bibfnamefont {Iordanis}\
  \bibnamefont {Kerenidis}}\ and\ \bibinfo {author} {\bibfnamefont {Anupam}\
  \bibnamefont {Prakash}},\ }\href@noop {} {\enquote {\bibinfo {title} {Quantum
  recommendation systems},}\ }\Eprint {http://arxiv.org/abs/arXiv:1603.08675}
  {arXiv:1603.08675} \BibitemShut {NoStop}%
\bibitem [{\citenamefont {Kerenidis}\ \emph {et~al.}(2020)\citenamefont
  {Kerenidis}, \citenamefont {Landman},\ and\ \citenamefont
  {Prakash}}]{Kerenidis2020Quantum}%
  \BibitemOpen
  \bibfield  {author} {\bibinfo {author} {\bibfnamefont {Iordanis}\
  \bibnamefont {Kerenidis}}, \bibinfo {author} {\bibfnamefont {Jonas}\
  \bibnamefont {Landman}}, \ and\ \bibinfo {author} {\bibfnamefont {Anupam}\
  \bibnamefont {Prakash}},\ }\bibfield  {title} {\enquote {\bibinfo {title}
  {Quantum algorithms for deep convolutional neural networks},}\ }in\ \href
  {https://openreview.net/forum?id=Hygab1rKDS} {\emph {\bibinfo {booktitle}
  {International Conference on Learning Representations}}}\ (\bibinfo
  {publisher} {ICLR},\ \bibinfo {address} {Addis Ababa, Ethiopia},\ \bibinfo
  {year} {2020})\BibitemShut {NoStop}%
\bibitem [{\citenamefont {Mangasarian}(1994)}]{mangasarian1994nonlinear}%
  \BibitemOpen
  \bibfield  {author} {\bibinfo {author} {\bibfnamefont {Olvi~L}\ \bibnamefont
  {Mangasarian}},\ }\href@noop {} {\emph {\bibinfo {title} {Nonlinear
  programming}}}\ (\bibinfo  {publisher} {Society for Industrial and Applied
  Mathematics},\ \bibinfo {address} {Philadelphia},\ \bibinfo {year}
  {1994})\BibitemShut {NoStop}%
\bibitem [{\citenamefont {Reddy}(2014)}]{reddy2014introduction}%
  \BibitemOpen
  \bibfield  {author} {\bibinfo {author} {\bibfnamefont {Junuthula~Narasimha}\
  \bibnamefont {Reddy}},\ }\href@noop {} {\emph {\bibinfo {title} {An
  Introduction to Nonlinear Finite Element Analysis: with applications to heat
  transfer, fluid mechanics, and solid mechanics, 2nd ed.}}}\ (\bibinfo
  {publisher} {Oxford University Press},\ \bibinfo {address} {Oxford},\
  \bibinfo {year} {2014})\BibitemShut {NoStop}%
\bibitem [{\citenamefont {Verhulst}(2006)}]{verhulst2006nonlinear}%
  \BibitemOpen
  \bibfield  {author} {\bibinfo {author} {\bibfnamefont {Ferdinand}\
  \bibnamefont {Verhulst}},\ }\href@noop {} {\emph {\bibinfo {title} {Nonlinear
  differential equations and dynamical systems}}}\ (\bibinfo  {publisher}
  {Springer},\ \bibinfo {address} {Heidelberg},\ \bibinfo {year}
  {2006})\BibitemShut {NoStop}%
\bibitem [{\citenamefont {Kerner}(1981)}]{kerner1981universal}%
  \BibitemOpen
  \bibfield  {author} {\bibinfo {author} {\bibfnamefont {Edward~H.}\
  \bibnamefont {Kerner}},\ }\bibfield  {title} {\enquote {\bibinfo {title}
  {Universal formats for nonlinear ordinary differential systems},}\ }\href
  {\doibase 10.1063/1.525074} {\bibfield  {journal} {\bibinfo  {journal}
  {Journal of Mathematical Physics}\ }\textbf {\bibinfo {volume} {22}},\
  \bibinfo {pages} {1366--1371} (\bibinfo {year} {1981})}\BibitemShut {NoStop}%
\bibitem [{\citenamefont {Cormen}\ \emph {et~al.}(2022)\citenamefont {Cormen},
  \citenamefont {Leiserson}, \citenamefont {Rivest},\ and\ \citenamefont
  {Stein}}]{cormen2022introduction}%
  \BibitemOpen
  \bibfield  {author} {\bibinfo {author} {\bibfnamefont {Thomas~H}\
  \bibnamefont {Cormen}}, \bibinfo {author} {\bibfnamefont {Charles~E}\
  \bibnamefont {Leiserson}}, \bibinfo {author} {\bibfnamefont {Ronald~L}\
  \bibnamefont {Rivest}}, \ and\ \bibinfo {author} {\bibfnamefont {Clifford}\
  \bibnamefont {Stein}},\ }\href@noop {} {\emph {\bibinfo {title} {Introduction
  to algorithms}}}\ (\bibinfo  {publisher} {MIT press},\ \bibinfo {address}
  {Cambridge, MA},\ \bibinfo {year} {2022})\BibitemShut {NoStop}%
\bibitem [{\citenamefont {He}(1999)}]{he1999homotopy}%
  \BibitemOpen
  \bibfield  {author} {\bibinfo {author} {\bibfnamefont {Ji-Huan}\ \bibnamefont
  {He}},\ }\bibfield  {title} {\enquote {\bibinfo {title} {Homotopy
  perturbation technique},}\ }\href {\doibase
  https://doi.org/10.1016/S0045-7825(99)00018-3} {\bibfield  {journal}
  {\bibinfo  {journal} {Computer Methods in Applied Mechanics and Engineering}\
  }\textbf {\bibinfo {volume} {178}},\ \bibinfo {pages} {257--262} (\bibinfo
  {year} {1999})}\BibitemShut {NoStop}%
\bibitem [{\citenamefont {Babolian}\ \emph {et~al.}(2009)\citenamefont
  {Babolian}, \citenamefont {Azizi},\ and\ \citenamefont
  {Saeidian}}]{babolian2009some}%
  \BibitemOpen
  \bibfield  {author} {\bibinfo {author} {\bibfnamefont {E.}~\bibnamefont
  {Babolian}}, \bibinfo {author} {\bibfnamefont {A.}~\bibnamefont {Azizi}}, \
  and\ \bibinfo {author} {\bibfnamefont {J.}~\bibnamefont {Saeidian}},\
  }\bibfield  {title} {\enquote {\bibinfo {title} {Some notes on using the
  homotopy perturbation method for solving time-dependent differential
  equations},}\ }\href {\doibase https://doi.org/10.1016/j.mcm.2009.03.003}
  {\bibfield  {journal} {\bibinfo  {journal} {Math. Comput. Modell.}\ }\textbf
  {\bibinfo {volume} {50}},\ \bibinfo {pages} {213--224} (\bibinfo {year}
  {2009})}\BibitemShut {NoStop}%
\bibitem [{\citenamefont {Chakraverty}\ \emph {et~al.}(2019)\citenamefont
  {Chakraverty}, \citenamefont {Mahato}, \citenamefont {Karunakar},\ and\
  \citenamefont {Rao}}]{chakraverty2019advanced}%
  \BibitemOpen
  \bibfield  {author} {\bibinfo {author} {\bibfnamefont {Snehashish}\
  \bibnamefont {Chakraverty}}, \bibinfo {author} {\bibfnamefont {Nisha}\
  \bibnamefont {Mahato}}, \bibinfo {author} {\bibfnamefont {Perumandla}\
  \bibnamefont {Karunakar}}, \ and\ \bibinfo {author} {\bibfnamefont
  {Tharasi~Dilleswar}\ \bibnamefont {Rao}},\ }\href@noop {} {\emph {\bibinfo
  {title} {Advanced numerical and semi-analytical methods for differential
  equations}}}\ (\bibinfo  {publisher} {Wiley},\ \bibinfo {address} {Hoboken,
  NJ},\ \bibinfo {year} {2019})\BibitemShut {NoStop}%
\bibitem [{\citenamefont {Rheinboldt}(1974)}]{rheinboldt1974methods}%
  \BibitemOpen
  \bibfield  {author} {\bibinfo {author} {\bibfnamefont {Werner~C.}\
  \bibnamefont {Rheinboldt}},\ }\href@noop {} {\emph {\bibinfo {title} {Methods
  for Solving Systems of Nonlinear Equations}}}\ (\bibinfo  {publisher}
  {Society for Industrial and Applied Mathematics},\ \bibinfo {address}
  {Philadelphia},\ \bibinfo {year} {1974})\BibitemShut {NoStop}%
\bibitem [{\citenamefont {Shewchuk}\ \emph {et~al.}(1994)\citenamefont
  {Shewchuk} \emph {et~al.}}]{shewchuk1994introduction}%
  \BibitemOpen
  \bibfield  {author} {\bibinfo {author} {\bibfnamefont {Jonathan~Richard}\
  \bibnamefont {Shewchuk}} \emph {et~al.},\ }\href@noop {} {\enquote {\bibinfo
  {title} {An introduction to the conjugate gradient method without the
  agonizing pain},}\ } (\bibinfo {year} {1994})\BibitemShut {NoStop}%
\bibitem [{\citenamefont {Koshy}(2008)}]{koshy2008catalan}%
  \BibitemOpen
  \bibfield  {author} {\bibinfo {author} {\bibfnamefont {Thomas}\ \bibnamefont
  {Koshy}},\ }\href@noop {} {\emph {\bibinfo {title} {Catalan numbers with
  applications}}}\ (\bibinfo  {publisher} {Oxford University Press},\ \bibinfo
  {address} {Oxford},\ \bibinfo {year} {2008})\BibitemShut {NoStop}%
\bibitem [{\citenamefont {Brassard}\ \emph {et~al.}(2002)\citenamefont
  {Brassard}, \citenamefont {Hoyer}, \citenamefont {Mosca},\ and\ \citenamefont
  {Tapp}}]{brassard2002quantum}%
  \BibitemOpen
  \bibfield  {author} {\bibinfo {author} {\bibfnamefont {Gilles}\ \bibnamefont
  {Brassard}}, \bibinfo {author} {\bibfnamefont {Peter}\ \bibnamefont {Hoyer}},
  \bibinfo {author} {\bibfnamefont {Michele}\ \bibnamefont {Mosca}}, \ and\
  \bibinfo {author} {\bibfnamefont {Alain}\ \bibnamefont {Tapp}},\ }\bibfield
  {title} {\enquote {\bibinfo {title} {Quantum amplitude amplification and
  estimation},}\ }\href@noop {} {\bibfield  {journal} {\bibinfo  {journal}
  {Contemp. Math.}\ }\textbf {\bibinfo {volume} {305}},\ \bibinfo {pages}
  {53--74} (\bibinfo {year} {2002})}\BibitemShut {NoStop}%
\end{thebibliography}%

\end{document}